\documentclass[11pt]{article}
\usepackage{fullpage}

\usepackage{graphics}
\usepackage[dvips]{epsfig}

\usepackage{amsmath}
\usepackage{amssymb}
\usepackage{amsfonts}
\usepackage{graphicx}

\usepackage{algorithm} 
\usepackage{algpseudocode}
\usepackage{amsthm}
\usepackage{color}

\newtheorem{theorem}{Theorem}[section]
\newtheorem{lemma}{Lemma}[section]

\newtheorem{claim}{Claim}[section]
\newtheorem{proposition}{Proposition}[section]

\newcommand{\cD}{{\cal D}}

\newenvironment{proofclaim}{\noindent{\bf Proof of the claim:}}{\hfill$\star$}

\newcommand{\remove}[1]{}

\begin{document}

\baselineskip  0.2in 
\parskip     0.1in 
\parindent   0.0in 

\title{{\bf Deterministic Treasure Hunt in the Plane\\ with Angular Hints}}

\author{
S\'{e}bastien Bouchard\thanks{Sorbonne Universit\'e, CNRS, INRIA, LIP6, F-75005 Paris, France, E-mail: sebastien.bouchard@lip6.fr}
\and
Yoann Dieudonn\'{e}\thanks{
MIS Lab., Universit\'{e} de Picardie Jules Verne, France, E-mail: yoann.dieudonne@u-picardie.fr}
\and
Andrzej Pelc\thanks{D\'{e}partement d'informatique, Universit\'{e} du Qu\'{e}bec en Outaouais,
Gatineau, Qu\'{e}bec J8X 3X7,
Canada. E-mail: pelc@uqo.ca.
Supported in part by NSERC discovery grant 8136 -- 2013 
and by the Research Chair in Distributed Computing of
the Universit\'{e} du Qu\'{e}bec en Outaouais.}
\and
Franck Petit\thanks{Sorbonne Universit\'e, CNRS, INRIA, LIP6, F-75005 Paris, France, E-mail: Franck.Petit@lip6.fr}
}

\date{ }
\maketitle

\begin{abstract}
A mobile agent equipped with a compass and a measure of length has to find an inert treasure in the Euclidean plane. Both the agent and the treasure are modeled as points. In the beginning, the agent is at a distance at most $D>0$ from the treasure, but knows neither the distance nor any bound on it. Finding the treasure means getting at distance at most 1 from it. The agent makes a series of moves. Each of them consists in moving straight in a chosen direction at a chosen distance. In the beginning and after each move the agent gets a hint consisting of a positive angle smaller than $2\pi$ whose vertex is at the current position of the agent and within which the treasure is contained. We investigate the problem of how these hints permit the agent to lower the cost of finding the treasure, using a deterministic algorithm, where the cost is the worst-case total length of the agent's trajectory. It is well known that without any hint the optimal (worst case) cost is $\Theta(D^2)$. We show that if all angles given as hints are at most $\pi$, then the cost can be lowered to $O(D)$, which is optimal. If all angles are at most $\beta$,
where $\beta<2\pi$ is a constant unknown to the agent, then the cost is at most $O(D^{2-\epsilon})$, for some $\epsilon>0$.
For both these positive results we present deterministic algorithms achieving the above costs. Finally, if angles given as hints can be arbitrary, smaller than $2\pi$, then we show that cost $\Theta(D^2)$ cannot be beaten.

\vspace*{2ex}


\vspace*{4cm}
\end{abstract}

\thispagestyle{empty}
\setcounter{page}{0}
\pagebreak

\section{Introduction}

{\bf Motivation.}
A tourist visiting an unknown town wants to find her way to the train station or a skier lost on a slope wants to get back to the hotel.
Luckily, there are many people that can help. However, often they are not sure of the exact direction: when asked about it, they make a vague gesture with the arm
swinging around the direction to the target, accompanying the hint with the words ``somewhere there''.  In fact,  they show an angle containing the target.
Can such vague hints help the lost traveller to find the way to the target? The aim of the present paper is to answer this question.

{\bf The model and problem formulation.}
A mobile agent equipped with a compass and a measure of length has to find an inert treasure in the Euclidean plane. Both the agent and the treasure are modeled as points. In the beginning, the agent is at a distance at most $D>0$ from the treasure, but knows neither the distance nor any bound on it. Finding the treasure means getting at distance at most 1 from it. In applications, from such a distance the treasure can be seen. The agent makes a series of moves. Each of them consists in moving straight in a chosen direction at a chosen distance. In the beginning and after each move the agent gets a hint consisting of a positive angle smaller than $2\pi$ whose vertex is at the current position of the agent and within which the treasure is contained. We investigate the problem of how these hints permit the agent to lower the cost of finding the treasure, using a deterministic algorithm, where the cost is the worst-case total length of the agent's trajectory.
It is well known that the optimal cost of treasure hunt without hints is $\Theta(D^2)$. (The algorithm of cost $O(D^2)$ is to trace a spiral with jump 1 starting at the initial position of the agent, and the lower bound $\Omega(D^2)$ follows from Proposition \ref{lower} which establishes this lower bound even assuming arbitrarily large angles smaller than $2\pi$ given as hints.)

{\bf Our results.}
We show that if all angles given as hints are at most $\pi$, then the cost of treasure hunt can be lowered to $O(D)$, which is optimal. Our real challenge here is in the fact that hints can be angles of size  {\em exactly} $\pi$, in which case the design of a trajectory always leading to the treasure, while being cost-efficient in terms of traveled distance, is far from obvious.

 If all angles are at most $\beta$,
where $\beta<2\pi$ is a constant unknown to the agent, then we prove that the cost is at most $O(D^{2-\epsilon})$, for some $\epsilon>0$.
 Finally, we show that arbitrary angles smaller than $2\pi$ given as hints cannot be of significant help: using such hints the cost $\Theta(D^2)$
   cannot be beaten.
 
 For both our positive results we present deterministic algorithms achieving the above costs. Both algorithms work in phases ``assuming'' that the treasure is contained in increasing squares centered at the initial position of the agent.  The common principle behind both algorithms is to move the agent to strategically chosen points in the current square, depending on previously obtained hints, and sometimes perform exhaustive search of small rectangles from these points,
 in order to guarantee that the treasure is not there.
 This is done in such a way that, in a given phase, obtained hints together with small rectangles exhaustively searched, eliminate a sufficient area of the square assumed in the phase to eventually permit finding the treasure.
 
 In both algorithms, the points to which the agent travels and where it gets hints are chosen in a natural way, although very differently in each of the algorithms.
 The main difficulty is to prove that the distance travelled by the agent is within the promised cost. In the case of the first algorithm, it is possible to cheaply exclude large areas not containing the treasure, and thus find the treasure asymptotically optimally. For the second algorithm, the agent eliminates smaller areas at each time,
 due to less precise hints, and thus finding the treasure costs more.


{\bf Related work.}
The problem of treasure hunt, i.e.,  searching for an inert target by one or more mobile agents  was investigated under many different scenarios.
The environment where the treasure is hidden may be a graph or a plane, and the search may be deterministic or randomized.
An early paper \cite{BeNe} showed that the best competitive ratio for deterministic treasure hunt on a line is 9. In \cite{DFG} the authors generalized this problem, considering a model where, in addition to travel length, the cost includes a payment for every turn of the agent.
The book \cite{AG} surveys both the search for a fixed target and the related rendezvous problem, where the target and the finder are both mobile and
their role is symmetric: they both cooperate to meet. This book is concerned mostly with randomized search strategies. Randomized treasure hunt strategies for star search, where the target is on one of $m$ rays, are considered in \cite{KRT}. In \cite{MP,TSZ} the authors study relations between the problems of treasure hunt and rendezvous in graphs.  The authors of \cite{BCR} study the task of finding a fixed point on the line and in the grid, and initiate the study of the task
of searching for an unknown line in the plane. This research is continued, e.g., in \cite{JL,La2}. In \cite{SF} the authors concentrate on game-theoretic aspects of
the situation where multiple selfish pursuers compete to find a target, e.g., in a ring. The main result of \cite{La} is an optimal algorithm to sweep a plane in order to locate an unknown fixed target, where locating means to get the agent originating at point $O$ to a point $P$ such that the target is in the segment $OP$. In \cite{FHGTM} the authors consider the generalization of the search problem in the plane to the case of several searchers. Collective treasure hunt in the grid by several agents  with bounded memory is investigated in \cite{ELSUW,KLUW}. In \cite{BKR}, treasure hunt with randomly faulty hints is considered in tree networks.   
By contrast, the survey \cite{CHI} and the book \cite{BN} consider pursuit-evasion games, mostly on graphs, where pursuers try to catch a fugitive target trying to escape.

\section{Preliminaries}

Since for $D\leq 1$ treasure hunt is solved immediately, in the sequel we assume $D>1$. 
Since the agent has a compass, it can establish an orthogonal coordinate system with point $O$ with coordinates $(0,0)$ at its starting position, the $x$-axis going
East-West and the $y$-axis going North-South. Lines parallel to the  $x$-axis will be called horizontal, and lines parallel to the $y$-axis will be called vertical. When
the agent at a current point $a$ decides to go to a previously computed point $b$ (using a straight line), we describe this move simply as ``Go to $b$''. A hint given to 
the agent currently located at point $a$ is formally described as an ordered pair $(P_1,P_2)$ of half-lines originating at $a$ such that the angle clockwise from $P_1$
to $P_2$ (including $P_1$ and $P_2$) contains the treasure.

The line containing points $A$ and $B$ is denoted by $(AB)$.
A segment with extremities $A$ and $B$ is denoted by $[AB]$ and its length is denoted $|AB|$.
Throughout the paper, a polygon is defined as a closed polygon (i.e., together with the boundary). For a polygon $S$, we will denote by $\mathcal{B}(S)$ (resp. $\mathcal{I}(S)$) the boundary of $S$ (resp. the interior of $S$, i.e.,  the set $S\setminus\mathcal{B}(S)$). A rectangle is defined as a non-degenerate rectangle, i.e., 
with all sides of strictly positive length. A rectangle with vertices $A,B,C,D$ (in clockwise order) is denoted simply by $ABCD$. A rectangle is {\em straight} if one of its sides is  vertical.

In our algorithms we use the following procedure {\tt RectangleScan}$(R)$ whose aim is to traverse a closed rectangle $R$ (composed of the boundary and interior) with known coordinates, so that the agent initially situated at some point of $R$ gets at distance at most 1 from every point of it and returns to the starting point. We describe the procedure for a straight rectangle whose vertical side is not shorter than the horizontal side. The modification of the procedure for arbitrarily positioned rectangles is straightforward. Let the vertices of the rectangle $R$ be $A$, $B$, $C$ and $D$, where $A$ is the North-West vertex and the others are listed clockwise. Let $a$ be the point at which the agent starts the procedure.

The idea of the procedure is to go to vertex $A$, then make a snake-like movement in which consecutive vertical segments are separated by a distance 1,  and then go back to point $a$. The agent ignores all hints gotten during the execution of the procedure. 
Suppose that the horizontal side of $R$ has length $m$ and the vertical side has length $n$, with $n \geq m$. Let $k=\lfloor m \rfloor$. Let $a_0,a_1,\dots ,a_k$ be points on the North horizontal side of the rectangle, such that $a_0=A$ and the distance between consecutive points is 1. Let $b_0,b_1,\dots ,b_k$ be points on the South horizontal side of the rectangle, such that $b_0=D$ and the distance between consecutive points is~1.

The pseudocode of procedure {\tt RectangleScan}$(R)$ is given in Algorithm \ref{alg:rc}.

\begin{algorithm}
				\caption{Procedure {\tt RectangleScan}($R$)}
				\label{alg:rc}
				\begin{footnotesize}
				\begin{algorithmic}[1]
				\If{$k$ is odd}
					\For{$i=0$ {\bf to $k-1$} {\bf step} 2}
						\State Go to $a_i$; Go to $b_i$; 
						\State Go to $b_{i+1}$; Go to $a_{i+1}$
					\EndFor
					\State Go to $a$
				\Else
					\For{$i=0$ {\bf to $k-2$} {\bf step} 2}
						\State Go to $a_i$; Go to $b_i$; 
						\State Go to $b_{i+1}$; Go to $a_{i+1}$
					\EndFor
					\State Go to $a_k$; Go to $b_k$
					\State Go to $a$				
				\EndIf
				\end{algorithmic}
				\end{footnotesize}
			\end{algorithm}

%

\begin{proposition}\label{prelim}
For every point $p$ of the rectangle $R$, the agent is at distance at most 1 from $p$ at some time of the
execution of Procedure  {\tt RectangleScan}$(R)$.
The cost of the procedure is at most $5n \cdot \max(m,2)$, where $n \geq m$ are the lengths of the sides of the rectangle.
\end{proposition}

\begin{proof}
During the execution of Procedure  {\tt RectangleScan}$(R)$ the agent traverses all segments $[a_i,b_i]$, for $i=0,1,\dots,k$.
Every point of $R$ is at distance at most 1 from some point of this union. This proves the first assertion. 
The cost of vertical moves is upper bounded by $(m+1)n$, the cost of horizontal moves is upper bounded by $m$, and the cost of getting from $a$ to $A$ and of returning back to $a$ after the scan is upper bounded
by $2(m+n)$. Hence the total cost of  procedure {\tt RectangleScan}$(R)$ is at most
$(m+1)n+m+2(m+n)\leq mn+6n\leq 5n \cdot \max(m,2)$.
\end{proof}

\section{Angles at most $\pi$}

In this section we consider the case when all angles given as hints are at most $\pi$. Without loss of generality we can assume that they are all equal to $\pi$,
completing any smaller angle to $\pi$ in an arbitrary way: this makes the situation even harder for the agent, as hints become less precise. For such hints we show  Algorithm {\tt TreasureHunt1} that finds the treasure at cost $O(D)$. This is of course optimal, as the treasure can be at any point at distance at most $D$ from the starting point of the agent.

For angles of size $\pi$, every hint is in fact a half-plane whose boundary line $L$ contains the current location of the agent. For simplicity, we will code such a hint as
$(L,right)$ or $(L,left)$, whenever the line $L$ is not horizontal, depending on whether the indicated half-plane is to the right (i.e., East) or to the left (i.e., West) of $L$.
For any non-horizontal line $L$ this is non-ambiguous. Likewise, when $L$ is horizontal, we will code a hint as $(L,up)$ or $(L,down)$,  depending on whether the indicated half-plane is up (i.e., North) from $L$ or down (i.e., South) from $L$.

In view of the work on $\phi$-self-approaching curves (cf. \cite{AA}) we first note that there is a big difference of difficulty between obtaining our result in the case when angles given as hints are bounded by some angle $\phi_0$ strictly smaller than $\pi$ and when they are {\em at most} $\pi$, as we assume. A $\phi$-self-approaching curve is a planar oriented curve such that, for each point B on the curve, the rest of the curve lies inside a wedge of angle $\phi$ with apex in B. In \cite{AA}, the authors prove the following property of these curves: for every $\phi<\pi$ there exists a constant $c(\phi)$ such that the length of any $\phi$-self-approaching curve is at most $c(\phi)$ times the distance $D$ between its endpoints. Hence, for hints bounded by some angle $\phi_0$ strictly smaller than $\pi$, our result could possibly be derived from the existing literature: roughly speaking, the agent should follow a trajectory corresponding to any $\phi_0$-self-approaching curve to find the treasure at a cost linear in $D$. Even then, transforming the continuous scenario of self-approaching curves to our discrete scenario presents some difficulties. However, the crucial problem is this: the constant
$c(\phi)$ from \cite{AA} diverges to infinity as $\phi$ approaches $\pi$, hence the result from \cite{AA} cannot be used when hints are arbitrary angles smaller than $\pi$. Moreover,
the result of \cite{AA} holds only when $\phi<\pi$ (the authors also emphasize that for each $\phi\geq\pi$, the property is false), and thus the above derivation is no longer possible for our purpose when $\phi=\pi$. Actually, this is the real difficulty of our problem: handling angles equal to $\pi$, i.e., half-planes.

We further observe that a rather straightforward treasure hunt algorithm of cost $O(D\log D)$, for hints being angles of size $\pi$, can be obtained using an immediate corollary of a theorem proven in \cite{G} by Gr\"unbaum: each line passing through the centroid of a convex polygon cuts the polygon into two convex polygons with areas differing by a factor of at most $\frac{5}{4}$. Suppose for simplicity that $D$ is known. Starting from the square of side length $2D$, centered at the initial position of the agent, this permits to reduce the search area from $P$ to at most $\frac{5P}{9}$ in a single move. Hence, after $O(\log D)$ moves, the search area is small enough to be exhaustively searched by 
procedure {\tt RectangleScan} at cost $O(D)$. However, the cost of each move during the reduction is not under control and can be only bounded by a constant multiple of $D$, thus giving the total cost bound $O(D\log D)$. By contrast, our algorithm controls both the remaining search area and the cost incurred in each move, yielding the optimal cost $O(D)$.

\subsection{High level idea of the algorithm}

In Algorithm {\tt TreasureHunt1} the agent acts in phases $j=1,2,3,\ldots$ where in each phase $j$ the agent ``supposes'' that the treasure is in a straight square $R_j$ centered at the initial position
of the agent, and of side length $2^j$. When executing a phase $j$, the agent successively moves to distinct points with the aim of using the hints at these points to narrow the search area that initially corresponds to $R_j$. In our algorithm, this narrowing is made in such a way that the remaining search area is always a straight rectangle. Often this straight rectangle is a strict superset of the intersection of all hints that the agent was given previously. This would seem to be a waste, as we are searching some areas that have been previously excluded. However, this loss is compensated by the ease of searching description and subsequent analysis of the algorithm, due to the fact that, at each stage, the search area is very regular.   

During a phase, the agent proceeds to successive reductions of the search area by moving to distinct locations, until it obtains a rectangular search area that is small enough to be searched directly at low cost using procedure {\tt RectangleScan}. In our algorithm, such a final execution of {\tt RectangleScan} in a phase is triggered as soon as the rectangle has a side smaller than $4$. If the treasure is not found by the end of this execution of procedure {\tt RectangleScan}, the agent learns that the treasure cannot be in the supposed straight square $R_j$ and starts the next phase from scratch by forgetting all previously received hints. This forgetting again simplifies subsequent analysis. The algorithm terminates at the latest by the end of phase $j_0=\lceil \log_2 D\rceil+1$, in which the supposed straight square $R_{j_0}$ is large enough to contain the treasure. 
Hence, if the cost of a phase $j$ is linear in $2^j$, then the cost of the overall solution is linear in the distance $D$.

In order to give the reader deeper insights in the reasons why our solution is valid and has linear cost, we need to give more precise explanations on how the search area  is reduced during a given phase $j\geq2$ (when $j=1$, the agent makes no reduction and directly scans the small search area using procedure {\tt RectangleScan}). Suppose that in phase $j\geq2$ the agent is at the center $p$ of a search area corresponding to a straight rectangle $R$, every side of which has length between $4$ and $2^j$ (note that this is the case at the beginning of the phase), and denote by $A,B,C$ and $D$ the vertices of $R$ starting from the top left corner and going clockwise. In order to reduce rectangle $R$, the agent uses the hint at point $p$. The obtained hint denoted by $(L_1,x_1)$ can be of two types: either a {\em good} hint or a {\em bad} hint. A good hint is a hint whose line $L_1$ divides one of the sides of $R$ into two segments such that the length $y$ of the smaller one is at least $1$. A bad hint is a hint that is not good.

If the received hint $(L_1,x_1)$ is good, then the agent narrows the search area to a rectangle $R'\subset R$ having the following three properties:
\begin{enumerate}
\item $R\setminus R'$ does not contain the treasure.
\item The difference between the perimeters of $R$ and $R'$ is $2y\geq2$.
\item The distance from $p$ to the center of $R'$ is exactly $\frac{y}{2}$.
\end{enumerate}
and then moves to the center of $R'$.

An illustration of such a reduction is depicted in Figure~\ref{fig:gb}(a). The reduced search area $R'$ is the rectangle $ABde$. 

\begin{figure}[!htbp]
\begin{center}
  \begin{minipage}[t]{0.49\linewidth}
    \centering
	\includegraphics[width=1\textwidth]{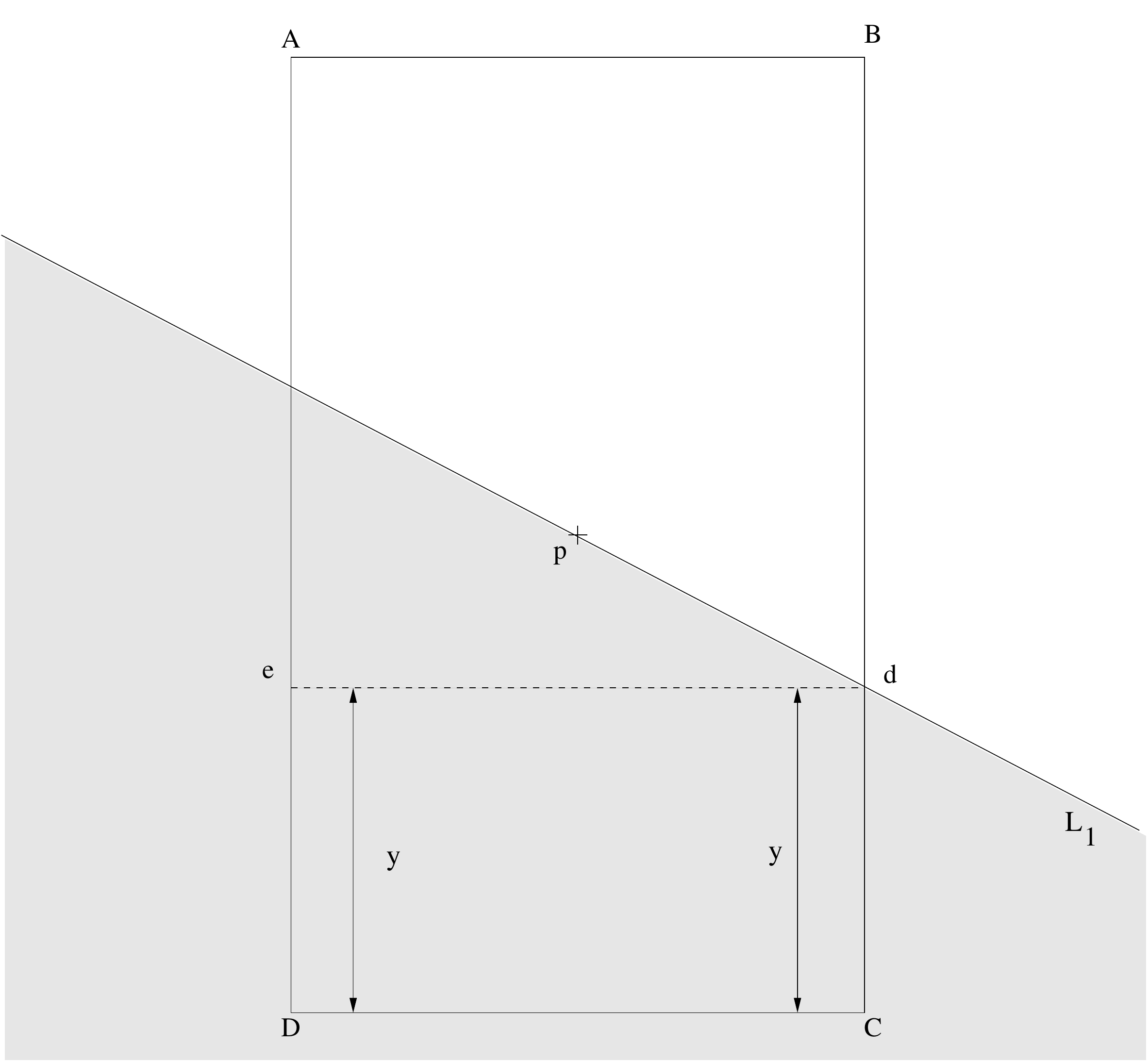}\\
    {\footnotesize ($a$) A good hint $(L_1,right)$}
  \end{minipage}
  \begin{minipage}[t]{0.49\linewidth}
    \centering
	\includegraphics[width=0.65\textwidth]{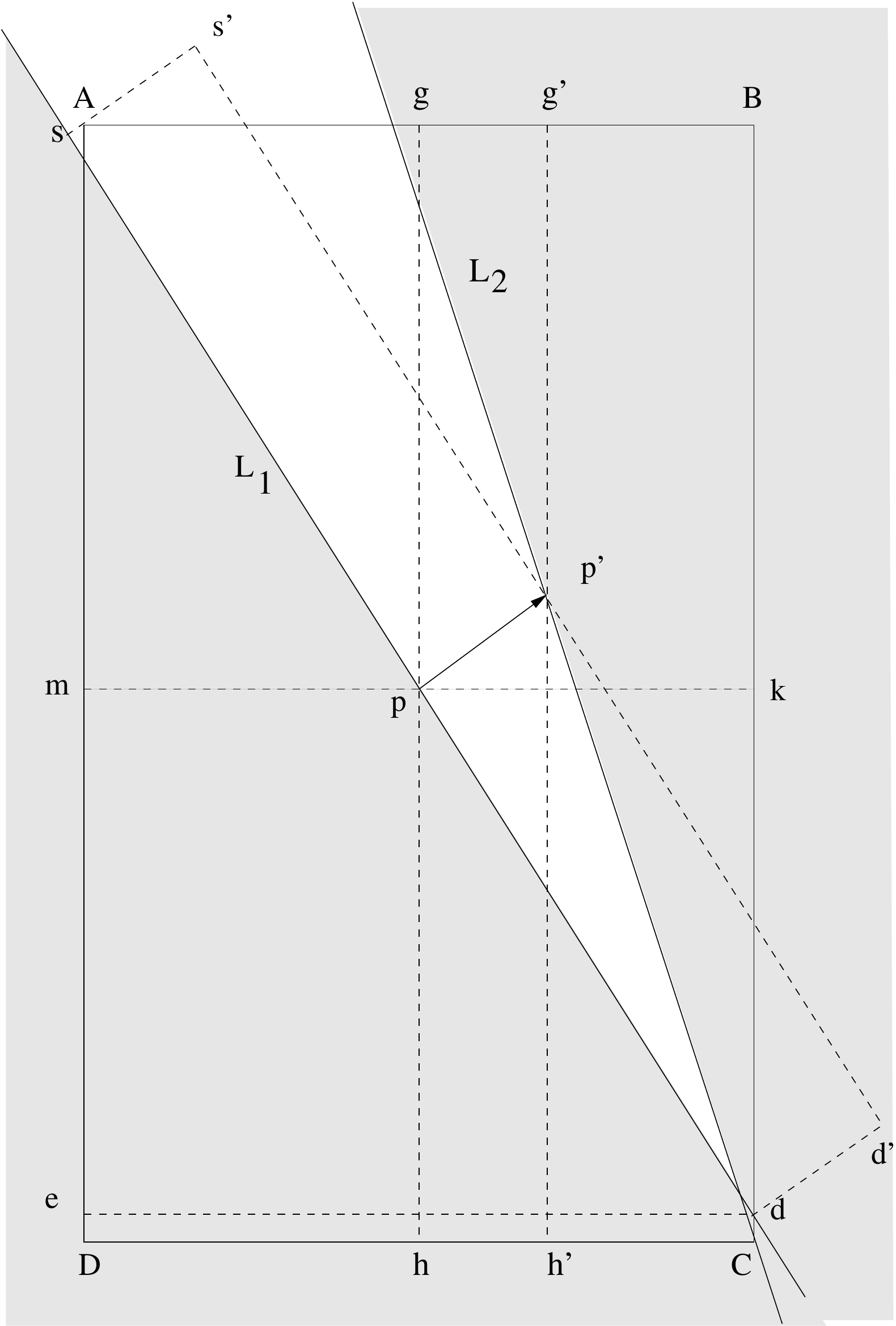}\\
    {\footnotesize ($b$) A bad hint $(L_1,right)$}
  \end{minipage}
\end{center}
 \caption{ In Figure (a)  the agent received a good hint  $(L_1,right)$ at the point $p$ of a rectangular search area $ABCD$. In Figure (b) it received  a bad hint
 $(L_1,right)$ at the point $p$ and hence it moved to point $p'$ and got a hint  $(L_2,left)$. In both figures the excluded half-planes are shaded.}
\label{fig:gb}
\end{figure}


If the agent receives a bad hint, say $(L_1,right)$, at the center of a rectangular search area $R$, we cannot apply the same method as the one used for a good hint: this is the reason for the distinction between good and bad hints. If we applied the same method as before, we could obtain a rectangular search area $R'$ such that the difference between the perimeters of $R$ and $R'$ is at least $2y$. However, in the context of a bad hint, the difference $2y$ may be very small (even null), and hence there is no significant reduction of the search area.  In order to tackle this problem, when getting a bad hint at the center $p$ of $R$, the agent moves to another point $p'$ which is situated in the half-plane $(L_1,right)$ at distance 2 from $p$, perpendicularly to $L_1$. This point $p'$ is chosen in such a way that, regardless of what is the second hint, we can ensure that two important properties described below are satisfied.

The first property is that by combining the two hints, the agent can decrease the search area to a rectangle $R'\subset R$ whose perimeter is smaller by $2$ compared to the perimeter of $R$, as it is the case for a good hint, and such that $R\setminus R'$ does not contain the treasure. This decrease follows either directly from the pair of hints, or indirectly after having scanned some relatively small rectangles using procedure {\tt RectangleScan}. In the example depicted in Fig. \ref{fig:gb} (b), after getting the second hint $(L_2,left)$,  the agent executes procedure 
{\tt RectangleScan}$(ss'd'd)$ followed by {\tt RectangleScan}$(gg'h'h)$ and moves to the center of the new search area $R'$ that is the rectangle $Agpm$. Note that the part of $R'$ not excluded by the two hints
and by the procedure 
{\tt RectangleScan} executed in rectangles $ss'd'd$ and  $gg'h'h$ is only the small quadrilateral bounded by line $L_2$ and the segments $[AB]$, $[s'd']$ and $[gh]$. However, in order to preserve the homogeneity of the process, we consider the entire new search area
$R'$ which is a straight rectangle whose perimeter is smaller by at least 2, compared to that from $R$. This follows from the fact that no side of $R$ has length smaller than 4. 
The agent finally moves to the center of $R'$. 

The second property is that all of this (i.e., the move from $p$ to $p'$, the possible scans of small rectangles and finally the move to the center of $R'$) is done at a cost linear in the difference of perimeters of $R$ and $R'$, as shown in Lemma~\ref{lem:rr}. The two properties together ensure that, even with bad hints, the agent manages to reduce the search area in a significant way and at a small cost.
So, regardless of whether hints are good or not, we can show that the cost of phase $j$ is in $\mathcal{O}(2^j)$ and the treasure is found during this phase if the initial square is large enough.
The difficulty of the solution is in showing that the moves prescribed by our algorithm in the case of bad hints guarantee the two above properties, and thus ensure the correctness of the algorithm and the cost linear in $D$.


\subsection{Algorithm and analysis}

In this subsection we describe our algorithm in detail, prove its correctness and analyze its complexity. Due to many possible positions of the line $L$ from the hint
$(L,x)$ obtained by the agent (the line $L$ cutting horizontal or vertical sides of the current search area, the slope of $L$ being positive or negative, and $x$ being $right$,
$left$, $up$ or $down$), there are many cases that the algorithm should consider. However, many of these cases can be treated similarly to one another, due to symmetry considerations. Hence, in order to reduce the number of cases,  we introduce some geometric transformations that enable us to consider only one representative case in each class. This case will be called a basic configuration.

We define a {\em configuration} as a couple $(R,(L,x))$, where $R$ is a straight rectangle,
and $(L,x)$ is a hint, i.e., a half-plane such that the line $L$ contains the center of $R$.

A configuration  $(R,(L,x))$ is called {\em lying} iff the line $L$ passes through a point that is in the interior of a vertical side of $R$. A configuration that is not lying is called {\em standing}.
A configuration  $(R,(L,x))$ is called {\em perfect} iff $L$ is horizontal or vertical. A configuration that is not perfect is called {\em imperfect}.

A perfect lying (resp. standing) configuration $(R,(L,x))$ can be of two types:
\begin{itemize}
\item {\bf Type 1.} $x=up$ (resp. $x=left$)
\item {\bf Type 2.} $x=down$ ( resp. $x=right$)
\end{itemize}


An imperfect configuration $(R,(L,x))$ can be of four types:
\begin{itemize}
\item {\bf Type 1.} The slope of $L$ is negative and $x=right$
\item {\bf Type 2.} The slope of $L$ is negative and $x=left$
\item {\bf Type 3.} The slope of $L$ is positive and $x=right$
\item {\bf Type 4.} The slope of $L$ is positive and $x=left$
\end{itemize}

The following proposition follows immediately from the above definitions.

\begin{proposition}
\label{pro:plat1}
For every configuration, there exists a unique positive integer $i\leq4$ such that this configuration is a perfect or imperfect configuration of type $i$.\end{proposition}

A configuration $(R,(L,x))$ is called {\em critical} iff the line $L$ divides a side of $R$ into two parts such that the length of the smaller part is less than $1$ (possibly $0$).

We will denote by $Rot_{v,\alpha}$ the rotation by the angle $\alpha$ with center $v$, and by $Sym_P$ the axial symmetry with axis $P$. 

The set of all configurations is denoted by $\mathcal{C}$. Given a configuration $(R,(L,x))$, we denote by $r$ and $H$, respectively, the center of $R$ and the vertical line passing through $r$. For every $i\in\{0,1,2,3\}$, we define the following functions that are intuitively rotations and axial symmetries of configurations.

$\sigma_i: \mathcal{C} \rightarrow \mathcal{C}$ is defined by the formula
$\sigma_i((R,(L,x))) =(Rot_{r,\frac{i\pi}{2}}(R),Rot_{r,\frac{i\pi}{2}}((L,x)))$

$\rho: \mathcal{C} \rightarrow \mathcal{C}$ is defined by the formula
$\rho( (R,(L,x)))=(Sym_H(R),Sym_H((L,x)))$

Using the above functions, we now define the following eight {\em elementary transformations} $\phi_i: \mathcal{C} \rightarrow \mathcal{C}$, for $i\in\{0,\dots, 7\}$.

For $i\in\{0,1,2,3\}$, we have  $\phi_i((R,(L,x)))= \sigma_i((R,(L,x)))$.\\
For $i\in\{4,5,6,7\}$, we have  $\phi_i((R,(L,x)))=  \rho(\sigma_{i-4}((R,(L,x))))$.

We say that a configuration is {\em  basic} iff it is either a lying perfect configuration of type 1 or a lying imperfect configuration of type 1.

The following proposition asserts that from every configuration we can obtain a basic configuration by at least one of the elementary transformations. This follows directly from the definitions. 

\begin{proposition}
\label{pro:plat2}
For every configuration $(R,(L,x))$, there exists $i\in\{0,\ldots,7\}$ and a basic configuration $(R',(L',x'))$ such that $(R',(L',x'))=\phi_i((R,(L,x)))$ 
\end{proposition}

For every configuration, the elementary transformation with the smallest index $i$ for which the above proposition is true will be called the \emph{basic transformation} of this configuration.

Note that, by applying to a configuration $(R,(L,x))$ its basic transformation $\phi_k$ in order to obtain $(R',(L',x'))=\phi_k((R,(L,x)))$ , each point $s$ of $(L,x)$ is rotated and possibly symmetrically reflected to obtain a new point $s'$ in  $(L',x')$. By a slight abuse of notation we will write $s'= \phi_k(s)$ and $s= \phi^{-1}_k(s')$,  and, more generally, for any set of points $S$, we will write $S'= \phi_k(S)$ and $S= \phi^{-1}_k(S')$.

Algorithm \ref{alg:plat} gives a pseudo-code of our main algorithm. It uses the function {\tt ReduceRectangle} described in Algorithm \ref{alg:RR}  that is the key technical tool permitting the agent to reduce its search area. The agent interrupts the execution of Algorithm  \ref{alg:plat} as soon as it gets at distance 1 from the treasure, at which point it can ``see'' it and thus treasure hunt stops.

\begin{figure}[httb!]
	\begin{center}
	\includegraphics[width=0.5\textwidth]{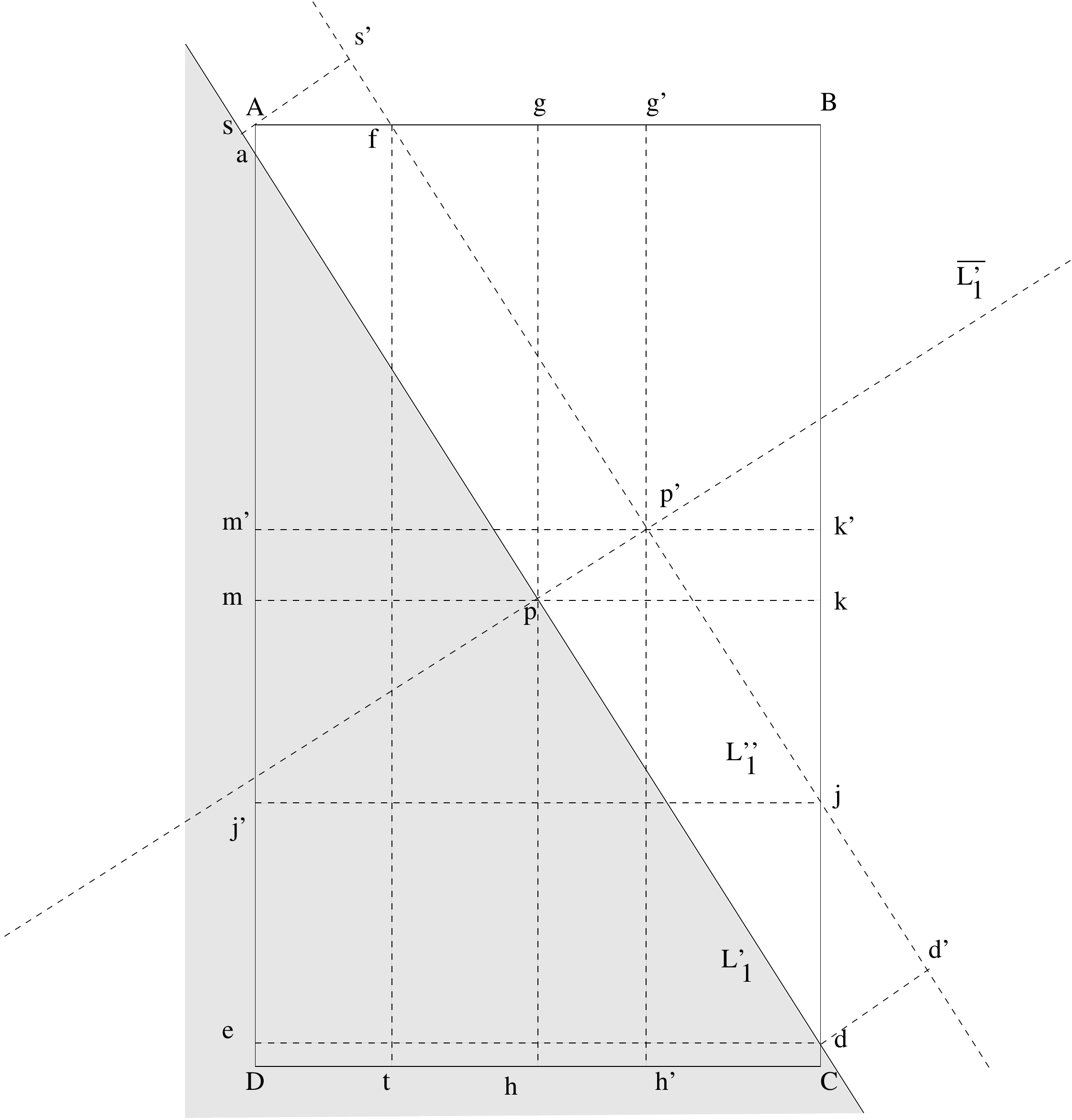}
	\caption{Illustration of the geometric objects used in Algorithm \ref{alg:RR} and in the proof of Lemma \ref{lem:rr}. We show an example of a basic configuration $(R',(L_1',x_1'))$ that is critical, in which $R'$ is the rectangle $ABCD$ and $x'_1=right$. We also show projections and intersections points defined in Algorithm \ref{alg:RR}.
	The excluded area is shaded.}
	\label{fig:f1}
	\end{center}
\end{figure}

\begin{algorithm}
	\caption{{\tt TreasureHunt1}}
	\label{alg:plat}
	\begin{footnotesize}
	\begin{algorithmic}[1]
		\State $O$:= the initial position of the agent
		\State $i$:=1
		\Loop 
    			\State $R_i$:= the straight square centered at $O$ with sides of length $2^i$\label{rr:1}
			\While{$R_i$ has no side with length smaller than $4$}\label{rr:condition}
				\State $R_i$:={\tt ReduceRectangle}($R_i$)\label{l:rr}
				\EndWhile\label{rr:2}
			\State Execute {\tt RectangleScan}($R_i$)\label{rr:3}
			\State Go to $O$\label{rr:4}
			\State $i$:=$i+1$
		\EndLoop	
	\end{algorithmic}
	\end{footnotesize}
\end{algorithm}

\begin{algorithm}
	\caption{Function {\tt ReduceRectangle}($R$)}
	\label{alg:RR}
		\begin{footnotesize}
	\begin{algorithmic}[1]
		\State $p:=$ the center of rectangle $R$		
		\State Let $(L_1,x_1)$ and $\phi_k$ be respectively the hint obtained at $p$ and the basic transformation of $(R,(L_1,x_1))$\label{l:basic}
		\State Let $(R',(L_1',x_1'))$ be the configuration such that $(R',(L_1',x_1'))=\phi_k((R,(L_1,x_1)))$\label{l:R'}
		\State Let $A$,$B$,$C$ and $D$ be the vertices of $R'$ in clockwise order, starting from the top-left corner\label{l:basicdef1}
		\State Let $a$ (resp. $d$) be the intersection between $L_1'$ and $(AD)$ (resp. $(BC)$)
		\State Let $e$ be the orthogonal projection of $d$ onto segment $[AD]$\label{l:basicdef2}	 
		\If{$(R',(L_1',x_1'))$ is not critical}
			\State NewRectangle:= the rectangle $ABde$
		\Else
			\State Let $\overline{L_1'}$ be the line that is  perpendicular to $L_1'$\label{l:consis1}
			\State Let $p'$ be the point at distance $2$ from $p$ in $\overline{L_1'}\cap(L_1',x')$\label{l:p'}
			\State Let $L_1''$ be the parallel line to $L_1'$ passing through $p'$
			\State Let $f$ (resp. $j$) be the intersection of $L_1''$ and segment $[AB]$ (resp. segment $[BC]$)
			\State Let $j'$ be the orthogonal projection of $j$ onto segment $[AD]$
			\State Let $t$ be the orthogonal projection of $f$ onto segment $[DC]$
			\State Let $m'$ (resp. $k'$) be the orthogonal projection of $p'$ onto segment $[AD]$ (resp. $[BC]$)
			\State Let $m$ (resp. $k$) be the orthogonal projection of $p$ onto segment $[AD]$ (resp. $[BC]$)
			\State Let $g'$ (resp. $h'$) be the orthogonal projection of $p'$ onto segment $[AB]$ (resp. $[DC]$)
			\State Let $g$ (resp. $h$) be the orthogonal projection of $p$ onto segment $[AB]$ (resp. $[DC]$)
			\State Let $s$ (resp. $s'$) be the orthogonal projection of $A$ onto line $L_1'$ (resp. $L_1''$)
			\State Let $d'$  be the orthogonal projection of $d$ onto line $L_1''$\label{l:consis2}
			\State Go to $\phi^{-1}_k(p')$ \label{l:move2}
			\State Let $(L_2,x_2)$ be the hint obtained at $\phi^{-1}_k(p')$ and let $(L'_2,x'_2)=\phi^{-1}_k((L_2,x_2))$
			\If{$x'_2=right$ and $L'_2$ is clockwise between $L_1''$ (included) and $(pp')$ (excluded)}\label{l:debuts}
			\State NewRectangle:= the rectangle $fBCt$
			\EndIf
			\If{$x'_2=right$ and $L'_2$ is clockwise between $(pp')$ (included) and $(m'k')$ (excluded)}
			\State Execute {\tt RectangleScan}($\phi^{-1}_k(m'k'km)$)
			\State NewRectangle:= the rectangle $gBCh$
			\EndIf

			\If{$x'_2$ $\in\{down,left\}$ and $L'_2$ is clockwise between $(m'k')$ (included) and $L_1''$ (excluded)}
			\State Execute {\tt RectangleScan}($\phi^{-1}_k(ss'd'd)$)
			\State Execute {\tt RectangleScan}($\phi^{-1}_k(m'k'km)$)
			\State NewRectangle:= the rectangle $pkCh$
			\EndIf 
			
			\If{$x'_2=left$ and $L'_2$ is clockwise between $L_1''$ (included) and $(g'h')$ (excluded)}
			\State Execute {\tt RectangleScan}($\phi^{-1}_k(ss'd'd)$)
			\State Execute {\tt RectangleScan}($\phi^{-1}_k(gg'h'h)$)
			\State NewRectangle:= the rectangle $Agpm$
			\EndIf 

			\If{($x'_2=left$ and $L'_2$ is clockwise between $(g'h')$ (included) and $(pp')$ (excluded)) {\bf or} ($x'_2=left$ and $L'_2$ is clockwise between $(pp')$ (included) and $(m'k')$ (excluded)) {\bf or} ($x'_2 \in \{up,right\}$ and $L'_2$ is clockwise between $(m'k')$ (included) and $(p'k)$ (excluded))}\label{l:toobig}
			\State Execute {\tt RectangleScan}($\phi^{-1}_k(gg'h'h)$)
			\State NewRectangle:= the rectangle $ABkm$
			\EndIf
			\If{$x'_2=right$ and $L'_2$ is clockwise between $(p'k)$ (included) and $L_1''$ (excluded)}\label{l:fins}
			\State NewRectangle:= the rectangle $ABjj'$
			\EndIf
		\EndIf
			\State Let $o'$ be the center of NewRectangle\label{l:movefin1}
			\State Go to $\phi^{-1}_k(o')$
			\State \Return  $\phi^{-1}_k(NewRectangle)$\label{l:finplat}		
	\end{algorithmic}
		\end{footnotesize}
\end{algorithm}


\pagebreak

We now proceed to the proof of correctness and the complexity analysis of our algorithm.
In the following lemma, for every rectangle $R$, the function $Perimeter(R)$ returns the perimeter of the rectangle $R$.

\begin{lemma}
\label{lem:rr}
Let $R$ be a straight rectangle with no side of length less than $4$. If the agent executes {\tt ReduceRectangle}($R$) from the center of $R$, then at the end of this execution the following properties are satisfied.

\begin{enumerate}
\item The function {\tt ReduceRectangle}($R$) returns a straight rectangle $Rec$, such that $Rec\subset R$, and either $R\setminus Rec$ does not contain the treasure or the agent has seen the treasure.
\item $Perimeter(R)-Perimeter(Rec)\geq2$.
\item The agent is at the center of rectangle $Rec$.
\item The agent travelled a distance of at most $21(Perimeter(R)-Perimeter(Rec))$ during the execution of {\tt ReduceRectangle}($R$). 
\end{enumerate}
 
\end{lemma}

\begin{proof}
Most of the geometric objects used in the proof are explicitely defined in Algorithm~\ref{alg:RR}: in particular, this is the case of intersections or orthogonal projections (e.g., those in lines~\ref{l:consis1} to~\ref{l:consis2}). All other necessary objects will be defined within the proof. For the notation, refer to Fig. \ref{fig:f1}.

Consider the execution of function {\tt ReduceRectangle}($R$) starting at the center $p$ of $R$, where $R$ is a straight rectangle with no sides of length less than $4$.  
Denote by $z$ the position of the treasure in $(L_1,x_1)$. We have $(R',(L_1',x_1'))=\phi_k((R,(L_1,x_1)))$. In view of Proposition~\ref{pro:plat2}, it is enough to prove that the following three properties hold when the agent executes the last line of Algorithm~\ref{alg:RR}.

\begin{itemize}
\item {\bf P1}. The variable $NewRectangle$ is set to a straight rectangle $R''$ such that  $R''\subset R'$, and either $R'\setminus R''$ does not contain $\phi_k(z)$ or the agent has seen the treasure.
\item {\bf P2}. The inequality $Perimeter(R')-Perimeter(R'')\geq2$ holds.
\item {\bf P3}. The agent travelled a distance of at most $21(Perimeter(R)-Perimeter(R''))$ during the execution of function {\tt ReduceRectangle}($R$).
\end{itemize}

We first prove the above properties when $(R',(L_1',x_1'))$ is a non-critical configuration. In this case, the variable $NewRectangle$ is set to the straight rectangle $ABde$. Note that the points defined in lines~\ref{l:basicdef1} to~\ref{l:basicdef2} (in particular  the points $A$,$B$, $d$ and $e$) exist and $ABde$ is a straight rectangle such that $ABde\subset R'$ in view of the fact that $(R',(L_1',x_1'))$ is a basic configuration. Moreover, since $z\in(L_1,x_1)$, we have $\phi_k(z)\in(L'_1,x'_1)$. However, $edCD\cap(L'_1,x'_1)\subset [de]$ and $R'\setminus ABde=edCD\setminus [de]$. So we have $(R'\setminus ABde)\cap(L'_1,x'_1)=\emptyset$ and Property {P1} is satisfied. Property {P2} also holds because $(R',(L_1',x_1'))$ is a basic configuration that is not critical. Indeed, in that case we know that the length $|Bd|\leq|BC|-1$, as $|dC|\geq1$. Hence, $|Ae|+|Bd|\leq|AD|+|BC|-2$, and thus $Perimeter(R')-Perimeter(ABde)\geq2$. It remains to prove Property P3. If we denote by $\Delta$ the difference  $|BC|-|Bd|$,  the distance from $\phi_k(p)=p$ to the center $o'$ of rectangle ABde is exactly $\frac{\Delta}{2}$. Moreover, the distance from $p$ to $\phi_k^{-1}(o')$ is also $\frac{\Delta}{2}$, as $\phi_k^{-1}$ is a distance-preserving transformation. As a result, since the only movement of the agent is from $p$ to $\phi_k^{-1}(o')$ and $\Delta=\frac{Perimeter(R')-Perimeter(ABde)}{2}$, when the agent executes the last line of Algorithm~\ref{alg:RR}, it has traveled a distance of $\frac{Perimeter(R')-Perimeter(ABde)}{4}$ during the execution of function {\tt ReduceRectangle}($R$).
Thus the lemma holds if $(R',(L_1',x_1'))$ is a non-critical configuration. 

Let us now consider the more difficult situation when $(R',(L_1',x_1'))$ is a critical configuration. In Algorithm~\ref{alg:RR}, this situation is handled by moving the agent to the point $\phi_k^{-1}(p')$ (cf. line~\ref{l:move2}) where $p'$ is the point defined at line~\ref{l:p'}, in order to get a second hint $(L_2,x_2)$ at $\phi_k^{-1}(p')$. We have six cases to consider depending on the nature of $(L_2,x_2)$. Similarly as for non-critical configurations, we do not study the six cases directly on $(L_2,x_2)$, but on $(L'_2,x'_2)$ instead, where $(L'_2,x'_2)$ is such that $(L_2',x_2')=\phi_k((L_2,x_2))$. Note that if the list of cases for $(L'_2,x'_2)$ covers all possible situations, and in each of those cases Properties {P1} to {P3} are  satisfied, then the lemma will be proven.

The six cases correspond to the six conditional statements that are in lines~\ref{l:debuts} to~\ref{l:fins} of Algorithm~\ref{alg:RR}. The fact that these cases
cover all possible situations follows from the fact that $(R',(L_1',x_1'))$ is a basic configuration, by Proposition~\ref{pro:plat2}, and from the fact that the objects defined in lines~\ref{l:consis1} to~\ref{l:consis2} of Algorithm~\ref{alg:RR} exist. In turn, the existence of these objects follows from the definition of $R'$ and of $(L_1',x_1')$ as well as from the following three claims (note that $R'$ has no side with length less than $4$, as $\phi_k$ is a distance-preserving transformation and $R$ has no side with length less than $4$).

\begin{claim}
\label{cl:plat3}
$(R',(L_1',x_1'))$ is an imperfect lying configuration of type $1$.
\end{claim}

\begin{proofclaim}
Since $(R',(L_1',x_1'))$ is basic, we just have to show that it is not a perfect lying configuration of type $1$. Suppose by contradiction that it is. So, $x_1'=up$ and line $L_1'$ divides the west vertical side $[AD]$ (resp. the east vertical side $[BC]$) of rectangle $R'$ into two parts of equal length. Since, $(R',(L_1',x_1'))$ is critical, each of these parts has length less than $1$. As a result, $|AD|$ (resp. $|BC|$) is smaller than $2$. This implies that $R'$ has a side with a length smaller than $4$, which is a contradiction and concludes the proof of the claim.
\end{proofclaim}

\begin{claim}
\label{cl:plat4}
The point $p'$ belongs to $\mathcal{I}(R')$.
\end{claim}

\begin{proofclaim}
Since $p$ is the center of rectangle $R'$ that has no side of length less than $4$, every point that is at distance at most $2$ from $p$, and which is not one of the four orthogonal projections of $p$ on the sides of $R'$, necessarily belongs to $\mathcal{I}(R')$. However, by Algorithm~\ref{alg:RR}, point $p'$ is at distance 2 from $p$
on a line perpendicular to $L_1'$ and that passes through $p$. Moreover, by Claim~\ref{cl:plat3}, $(R',(L_1',x_1'))$ is an imperfect vertical configuration of type $1$, and thus the slope of $L_1'$ is negative. Hence the claim holds.
\end{proofclaim}

\begin{claim}
\label{cl:plat5}
The line $L_1''$ divides the northern side $[AB]$ (resp. the east side $[BC]$) of $R'$ into two parts of positive length.
\end{claim}

\begin{proofclaim}
In view of Claim~\ref{cl:plat3} and the fact that $L_1''$ is a line parallel to $L_1'$ passing through $p'$ that is a point belonging to $\mathcal{I}(R')$ (cf. Claim~\ref{cl:plat4}), it follows that $L_1''$ divides the east side $[BC]$ of $R'$ into two parts of positive length. It also follows that  $L_1''$ intersects the northern side $[AB]$ or the west side $[AD]$ of $R'$.
So to prove the claim, it is enough to show that $L_1''$ cannot intersect $[AD]$ (i.e., cannot pass through any points of $[AD]$ including the corners $A$ and $D$). 
Assume by contradiction that it does. Since $L_1''\subset(L_1',x_1')$ and the distance from any point of $L_1'$ to any point of $L_1''$ is at least $|pp'|$, then according to the definition of $L_1'$ and $L_1''$, we know that the segment $[AD]\cap(L_1',x_1')$ has a length that is at least $|pp'|=2$. However, by Claim~\ref{cl:plat3} and the fact that $(R',(L_1',x_1'))$ is a critical configuration, the segment $[AD]\cap(L_1',x_1')$ has a length that is smaller than $1$, which is a contradiction and proves the claim.
\end{proofclaim}

Hence, since we have a list of six cases covering all possible situations, it is enough to show that Properties {P1} to {P3} are  satisfied in each case, in order to conclude the proof of the lemma. Before analyzing them, let us give another claim that will be useful in the sequel.

\begin{claim}
\label{cl:plat6}
The length of segment $[Af]$ (resp. $[jC]$) is at least $1$.
\end{claim}

\begin{proofclaim}
As mentioned previously, the distance from any point of $L_1'$ to any point of $L_1''$ is at least $|pp'|=2$. Hence,  $|af|\geq2$ and $|jd|\geq2$. Since $d\in[jC]$ and $|aA|<1$ (because the configuration is critical) and $[af]$ is the hypothenuse of the right rectangle $Afa$, the claim follows.
\end{proofclaim}

The fact that each object that is assigned to variable $NewRectangle$ or given as input parameter to procedure {\tt RectangleScan} is a rectangle, can be shown using the above claims. Moreover, from the definitions of intersections and projections given in Algorithm~\ref{alg:RR}, it follows that each time procedure {\tt RectangleScan} is called with an input parameter corresponding to a rectangle $X$, the agent is in the rectangle $X$ (this is necessary in order to obtain a correct execution of the procedure). In the rest of the proof we will not mention this fact. Similarly, it follows  
 that a rectangle that is assigned to variable NewRectangle is always straight. 
 
Now, we consider the six cases and we start with the first one in which {$x'_2=right$ and $L'_2$ is clockwise between $L_1''$ (included) and $(pp')$ (excluded)}. In this case, variable $NewRectangle$ is set to the straight rectangle $fBCt\subset R'$. Since $z\in(L_1,x_1)\cap(L_2,x_2)$, we have $\phi_k(z)\in(L'_1,x'_1)\cap(L'_2,x'_2)$. However, $R'\setminus fBCt=AftD\setminus[ft]$, and in view of the value of $x'_2$ and the position of $L'_2$, we have $AftD\cap(L'_1,x'_1)\cap(L'_2,x'_2)\subseteq\{f\}$ (more precisely, $AftD\cap(L'_1,x'_1)\cap(L'_2,x'_2)=\{f\}$ if $L'_2=L_1''$, and $AftD\cap(L'_1,x'_1)\cap(L'_2,x'_2)=\emptyset$ for all the other positions of $L_2'$ within the considered case). Hence, $(R'\setminus fBCt)\cap(L'_1,x'_1)\cap(L'_2,x'_2)=\emptyset$ and  Property {P1} is satisfied. Concerning Property {P2}, we know that $|fB|=|AB|-|Af|$, which implies $|fB|\leq|AB|-1$ because $|Af|\geq1$ according to Claim~\ref{cl:plat6}. So, $Perimeter(R')-Perimeter(fBCt)\geq2$, and thus Property {P2} holds. Concerning Property P3, we need to evaluate the distance travelled by the agent when it moves from $p$ to $\phi_k^{-1}(p')$, and then from $\phi_k^{-1}(p')$ to $\phi_k^{-1}(o')$ (where $o'$ is the center of rectangle $fBCt$). Note that the distance from $p$ to $\phi_k^{-1}(o')$ is $\frac{\Delta}{2}$ where $\Delta$ is the difference  $|AB|-|fB|$. Moreover, $|pp'|=|p\phi_k^{-1}(p')|=2$. Hence $|p\phi_k^{-1}(p')|\leq2\Delta$ because $|AB|-|fB|=|Af|$ and $|Af|\geq1$ according to Claim~\ref{cl:plat6}. Thus, moving from $p$ to $\phi_k^{-1}(p')$ makes the agent travel a distance of at most $2\Delta$. Moving from $\phi_k^{-1}(p')$ to $\phi_k^{-1}(o')$ makes the agent travel a distance that is upper bounded by $|\phi_k^{-1}(p')p|+|p\phi_k^{-1}(o')|\leq\frac{5\Delta}{2}$. As a result, the total distance traveled by the agent is at most $\frac{9\Delta}{2}=\frac{9(Perimeter(R')-Perimeter(fBCt))}{4}$, as $\Delta=\frac{Perimeter(R')-Perimeter(fBCt)}{2}$. Hence Properties P1, P2 and P3 hold in this case.

Let us now consider the situation when {$x'_2=right$ and $L'_2$ is clockwise between $(pp')$ (included) and $(m'k')$ (excluded)}. The variable $NewRectangle$ is then set to the straight rectangle $gBCh\subset R'$. Note that $R'\setminus gBCh\subset AghD$. In view of the value of $x'_2$ and the position of $L'_2$, $AghD\cap(L'_1,x'_1)\cap(L'_2,x'_2)$ is included in the rectangle $m'k'km$. Since $\phi_k(z)\in(L'_1,x'_1)\cap(L'_2,x'_2)$, if the agent has not seen the treasure after having executed {\tt RectangleScan}($\phi^{-1}_k(m'k'km)$), then in view of Proposition~\ref{prelim} and the definition of transformation $\phi_k$ we know that $\phi_k(z)$ cannot be in the rectangle $AghD$. Thus, Property {P1} is satisfied. Property P2 follows from the facts that $|gB|=\frac{|AB|}{2}$ (since $g$ is the orthogonal projection of the center $p$ of $R'$ on the top side $[AB]$ of $R'$) and that $|AB|\geq4$ (as $R'$ has no side of length less than $4$). So, it remains to check the validity of Property P3 in the current case. The move of the agent is composed of three parts: the first part is when it moves from $p$ to $\phi_k^{-1}(p')$, the second part corresponds to the move made when executing procedure {\tt RectangleScan}($\phi^{-1}_k(m'k'km)$), and the third part is when the agent moves to $\phi_k^{-1}(o')$ (where $o'$ is the center of the rectangle $gBCh$).
Note that the execution of {\tt RectangleScan}($\phi^{-1}_k(m'k'km)$) starts and finishes at point $\phi_k^{-1}(p')$. This implies that the third part corresponds precisely to a move from point $\phi_k^{-1}(p')$ to point $\phi_k^{-1}(o')$. So, by similar arguments to those used in the previous case, we can show that the distance traveled in the first part plus the distance traveled in the third part gives a total of at most $\frac{9\Delta}{2}$ where, in the current situation, $\Delta$ is the difference  $|AB|-|gB|$. For the second part, corresponding to the execution of procedure {\tt RectangleScan}($\phi^{-1}_k(m'k'km)$), note that since $|p\phi^{-1}_k(p')|=2$, we have $|kk'|\leq2$ because $k$ (resp. $k'$) is the orthogonal projection of $p$ (resp. $p'$) onto $[BC]$. Moreover, $|m'k'|=|AB|$, and in view of the definition of $R'$, we have $|AB|\geq4$. Hence, according to Proposition~\ref{prelim}, we know that the agent travels a distance of at most $10|AB|$ during the second part. So the total distance traveled by the agent is at most $\frac{9\Delta}{2}+10|AB|$. As explained for Property P2, we know that $|gB|=\frac{|AB|}{2}$. Hence, $\Delta=\frac{|AB|}{2}$, $Perimeter(R')-Perimeter(gBCh)=|AB|$, and thus the total distance traveled by the agent is at most $\frac{49(Perimeter(R')-Perimeter(gBCh))}{4}$. As a result, Properties P1, P2 and P3 hold in this case.

We continue by analyzing the situation when {$x'_2$ $\in\{down,left\}$ and $L'_2$ is clockwise between $(m'k')$ included and $L_1''$ (excluded)}. In this situation, variable $NewRectangle$ is set to the straight rectangle $pkCh\subset R'$. In view of the value of $x'_2$ and the position of $L'_2$, we have $(L'_1,x'_1)\cap(L'_2,x'_2)\cap(R'\setminus pkCh)\subset (ss'd'd \cup m'k'km)$. Since $\phi_k(z)\in(L'_1,x'_1)\cap(L'_2,x'_2)$, if the agent has not seen the treasure after having executed {\tt RectangleScan}($\phi^{-1}_k(ss'd'd)$) followed by {\tt RectangleScan}($\phi^{-1}_k(m'k'km)$), then in view of Proposition~\ref{prelim} and the definition of transformation $\phi_k$ we know that $\phi_k(z)$ cannot be in $R'\setminus pkCh$. Thus, Property {P1} is satisfied. We can show that Property P2 also holds by similar arguments to those used to show Property P2 in the previous case. Concerning Property P3,  note that the move of the agent can be divided into four parts: the first part is when it moves from $p$ to $\phi_k^{-1}(p')$, the second (resp. third) part corresponds to the move made when executing procedure {\tt RectangleScan}($\phi^{-1}_k(ss'd'd)$) (resp. {\tt RectangleScan}($\phi^{-1}_k(m'k'km)$)), and the fourth part is when the agent moves to $\phi_k^{-1}(o')$ (where $o'$ is here the center of rectangle $pkCh$). Note that the execution of {\tt RectangleScan}($\phi^{-1}_k(ss'd'd)$) (resp. {\tt RectangleScan}($\phi^{-1}_k(m'k'km)$)) starts and finishes at point $\phi_k^{-1}(p')$. So, the fourth part is actually a move from point $\phi_k^{-1}(p')$ to point $\phi_k^{-1}(o')$. It is worth mentioning that moving from $\phi_k^{-1}(p')$ to point $\phi_k^{-1}(o')$ costs the same or less than first moving from $\phi_k^{-1}(p')$ to  $p$, and then moving from $p$ to  $\phi_k^{-1}(o')$. Moreover, moving from $p$ to  $\phi_k^{-1}(o')$ costs at most $\frac{\Delta_1+\Delta_2}{2}$ where $\Delta_1$ (resp. $\Delta_2$) is the difference $|gh|-|pg|$ (resp. $|mk|-|pm|$). Hence during the fourth part, the agent travels a distance of at most $2+ \frac{\Delta_1+\Delta_2}{2}$. During the first part, the agent travels a distance $2$. What about the second and third parts? To evaluate these costs we need to evaluate the lengths and widths of rectangles $ss'd'd$ and $m'k'km$. In the analysis of the previous case, we have shown that the length and width of rectangle $m'k'km$ are respectively $|AB|$ and at most $2$. Concerning rectangle $ss'd'd$, we have the following claim.

\begin{claim}
\label{cl:plat7}
$|ss'|=2$ and $2<|sd|<1+|AC|$.
\end{claim}

\begin{proofclaim}
Note that $|ss'|$ is exactly $2$ because $s$ (resp. $s'$) is the orthogonal projection of the corner $A$ onto line $L_1'$ (resp. $L_1''$). Also note that $|sd|=|sa|+|ad|$ where $[sa]$ is a side of the right triangle $asA$ whose hypotenuse is $[Aa]$. However, by Claim~\ref{cl:plat3} and the fact that $(R',(L_1',x_1'))$ is critical, we know that $|Aa|<1$. Moreover, $[ad]\subset R'$ and $|ad|\geq|AB|\geq4$. Hence, $2<|sd|<1+|AC|$, which concludes the proof of the claim.
\end{proofclaim}

As a result, according to Proposition~\ref{prelim}, we know that the agent travels a distance of at most $10(1+|AC|)$ during the second part and a distance of at most $10|AB|$ during the third part. Hence, the total distance traveled by the agent is at most $2+\frac{\Delta_1+\Delta_2}{2}+10(|AB|+|AC|+1)$. Note that $|gh|-|pg|=\frac{|AD|}{2}$, $|mk|-|pm|=\frac{|AB|}{2}$ and $|AC|<|AB|+|AD|$. Furthermore, in view of the fact that $R'$ has no side of length less than $4$, we have $\frac{|AD|}{2}\geq2$ and $\frac{|AB|}{2}\geq2$. So, the total distance traveled by the agent is at most $\frac{|AD|}{2}+\frac{|AD|+|AB|}{4}+10(2|AB|+|AD|+1)$. This in turn gives us a traveled distance that is upper bounded by $21(Perimeter(R')-Perimeter(pkCh))$ as $Perimeter(R')-Perimeter(pkCh)=2(\frac{|AD|}{2}+\frac{|AB|}{2})=|AD|+|AB|$. 
Consequently, Properties P1, P2 and P3 are valid in this case.

So far, we have analyzed the first three cases among the six cases that permit us to cover entirely the situation when $(R',(L_1',x_1'))$ is a critical configuration. However,  the arguments we need to use in order to analyze the last three cases are similar to those already used to analyze the first three cases. In particular, this is true for the fourth case when {$x'_2=left$ and $L'_2$ is clockwise between $L_1''$ (included) and $(g'h')$ (excluded)}: using a similar reasoning to that for the third case we have analyzed just above, we can show that Properties P1, P2 and P3 are also valid here. For the fifth case, which corresponds to the boolean expression of line~\ref{l:toobig}, Properties P2 and P3 can be proven using a similar reasoning to that used above to prove Properties P2 and P3 when {$x'_2=right$ and $L'_2$ is clockwise between $(pp')$ (included) and $(m'k')$ (excluded)}. Concerning Property P1, note that variable $NewRectangle$ is set here to the straight rectangle $ABkm\subset R'$. In view of the possible values of $x'_2$ and the possible positions of $L'_2$, we have $(L'_1,x'_1)\cap(L'_2,x'_2)\cap(R'\setminus ABkm)\subset gg'h'h$ if $x'_2=left$ and $L'_2$ is clockwise between $(g'h')$ (included) and $(pp')$ (excluded). Otherwise, we have $(L'_1,x'_1)\cap(L'_2,x'_2)\cap(R'\setminus ABkm)=\emptyset$. Since $\phi_k(z)\in(L'_1,x'_1)\cap(L'_2,x'_2)$, if the agent has not seen the treasure after having executed {\tt RectangleScan}($\phi^{-1}_k(gg'h'h)$),  then in view of Proposition~\ref{prelim} and the definition of transformation $\phi_k$ we know that $\phi_k(z)$ cannot be in $R'\setminus ABkm$. Thus  Property {P1} is also true in this case. Finally, the fact that Properties P1, P2 and P3 are true in the last of the six cases i.e., when {$x'_2=right$ and $L'_2$ is clockwise between $(p'k)$ (included) and $L_1''$ (excluded)} can be proven using similar arguments to those used for the first case, i.e.,  when {$x'_2=right$ and $L'_2$ is clockwise between $L_1''$ (included) and $(pp')$ (excluded)}. This completes the proof of the lemma.
\end{proof}

\begin{theorem}
Consider an agent $A$ and a treasure located at distance at most $D$ from the initial position of $A$. By executing Algorithm {\tt TreasureHunt1}, agent $A$ finds the treasure after having traveled a distance $\mathcal{O}(D)$. 
\end{theorem}

\begin{proof}
The execution of Algorithm~\ref{alg:plat} can be divided into phases $1,2,3,\ldots$ where phase $j\geq1$ is the part of the execution in which variable $i$ of Algorithm~\ref{alg:plat} is equal to $j$. 

In view of the second and third properties of Lemma~\ref{lem:rr} and lines~\ref{rr:1} to~\ref{rr:2} of Algorithm~\ref{alg:plat}, the number of calls to function {\tt ReduceRectangle} is bounded by the perimeter of a square with side length $2^j$. Hence we have the following claim.

\begin{claim}
\label{cl:plat8}
For every $j\geq1$, the number of calls to function {\tt ReduceRectangle}, within phase $j$, is bounded by $2^{j+2}$.
\end{claim}

In order to conclude the proof of the theorem, it is enough to prove the following two statements:
\begin{enumerate}
\item
for all $j\geq1$, the following property $\mathcal{H}_j$ holds:\\
 at the beginning of phase $j$ the agent has traveled a distance of at most $2^{j+7}$. 
\item 
the agent finds the treasure before starting phase $\lceil \log_2 D \rceil +2$. 
\end{enumerate}
We start by proving the first statement by induction on $j$. Note that property $\mathcal{H}_1$ is true because at the beginning of phase $1$ the agent has traveled a distance $0$. So, assume that, for a positive integer $\lambda$, property $\mathcal{H}_{\lambda}$ is true. We prove that property $\mathcal{H}_{\lambda+1}$ is also true.
Within phase $\lambda$, the move of the agent can be divided into two parts: the first part corresponds to the moves made when executing lines~\ref{rr:1} to~\ref{rr:2} of Algorithm~\ref{alg:plat}, while the second part corresponds to the moves made when executing lines~\ref{rr:3} and~\ref{rr:4} of Algorithm~\ref{alg:plat}.
By Claim~\ref{cl:plat8}, we know that the number $\tau$ of calls to function {\tt ReduceRectangle} during phase $\lambda$ is upper bounded by $2^{\lambda+2}$. For all $1\leq s\leq \tau$, we denote by $Q_s$ (resp. $Q'_s$) the rectangle that is the input parameter (resp. the returned value) of the $s$th call to function {\tt ReduceRectangle} during phase $\lambda$. Note that, for all $2\leq s\leq \tau$, $Q_s=Q'_{s-1}$. So, by the fourth property of Lemma~\ref{lem:rr}, the distance traveled by the agent during the first part of phase $\lambda$ is upperbounded by

\begin{align}
\label{platform:1}
~&21\sum_{s=1}^{s=\tau}(Perimeter(Q_s)-Perimeter(Q'_s))\\
~&\leq 21*(Perimeter(Q_1)-Perimeter(Q'_\tau)+\sum_{s=2}^{s=\tau}(Perimeter(Q_s)-Perimeter(Q'_{s-1})))\\
~&\leq 21*(Perimeter(Q_1)-Perimeter(Q'_\tau))~~~~~~~~~\mbox{(because  for all $2\leq s\leq \tau$, $Q_s=Q'_{s-1}$)}\\
~&\leq 21*Perimeter(Q_1)=21*2^{\lambda+2} \label{platform:4}
\end{align}

Concerning the second part of phase $\lambda$, it is worth mentioning that when the agent starts executing line~\ref{rr:3} of Algorithm~\ref{alg:plat}, variable $R_i$ is set  to a straight rectangle whose at least one side has length smaller than $4$ (according to line~\ref{rr:condition}),  and no sides have length larger than $2^{\lambda}$: indeed, using the first property of Lemma~\ref{lem:rr}, it follows by induction on $s$ that the straight rectangle $Q'_s$ is included in the straight rectangle $Q_1$, for all $1\leq s \leq \tau$. Moreover, the distance between any two points of $Q_1$ (and thus the cost of line \ref{rr:4} of Algorithm~\ref{alg:plat}) is at most $2^{\lambda +1}$.
Hence, in view of Proposition~\ref{prelim}, we know that the distance traveled by the agent during the second part of phase $\lambda$ is upper bounded by $22\cdot 2^{\lambda}$. From this and (\ref{platform:4}), we know that the total distance traveled during phase $\lambda$ is at most $2^{\lambda+7}$. Moreover, by the inductive hypothesis,  $\mathcal{H}_{\lambda}$ is true i.e., at the beginning of phase $\lambda$ the agent has traveled a distance of at most $2^{\lambda+7}$. As a result, when starting phase $\lambda+1$, the agent has traveled a total distance of at most $2^{\lambda+8}$. Thus, property $\mathcal{H}_{\lambda+1}$ is true, which concludes the inductive proof and thus proves the validity of the first statement.

Now let us focus on the second statement: the agent finds the treasure before starting phase $\lceil \log_2 D \rceil +2$. Suppose by contradiction that this is not the case. By Claim~\ref{cl:plat8} and Lemma~\ref{lem:rr}, at some point the agent starts executing phase $\lceil \log_2 D \rceil +1$. In view of Algorithm~\ref{alg:plat}, when the agent finishes the execution of line~\ref{rr:1} in phase $\lceil \log_2 D \rceil +1$, the value of variable $R_i$ is a square $S$ containing the treasure: indeed this square is centered at the initial position $O$ of the agent and it contains all points at distance at most $D$ from $O$ because its side length is $2^{\lceil \log_2 D \rceil +1}\geq 2D$, since $D>1$. 

Denote by $Q_{final}$ the rectangle returned by the last call to function {\tt ReduceRectangle} in phase $\lceil \log_2 D \rceil +1$: since the side length of $S$ is at least $2^{\lceil \log_2 D \rceil +1}\geq2^2$, this rectangle exists because the agent executes at least once line~\ref{l:rr} of Algorithm~\ref{alg:plat}. By Claim~\ref{cl:plat8} and Lemma~\ref{lem:rr}, at some point the agent executes line~\ref{rr:3} of Algorithm~\ref{alg:plat} and when the agent starts executing this line we know that it is at the center
of $Q_{final}$. Moreover, from Lemma~\ref{lem:rr}, it follows by induction on the number of calls to function {\tt ReduceRectangle} within phase $\lceil \log_2 D \rceil +1$, that the treasure does not belong to $S\setminus Q_{final}$, as otherwise the agent would have found the treasure before starting phase $\lceil \log_2 D \rceil +2$ which would be a contradiction. However, the treasure belongs to square $S$. Hence, the treasure belongs to $Q_{final}$, and by applying procedure {\tt RectangleScan}($Q_{final}$) (cf. line~\ref{rr:3}) from the center of $Q_{final}$, the agent necessarily finds the treasure by the end of the execution of this procedure, and thus by the end of phase $\lceil \log_2 D \rceil +1$. This gives a contradiction that proves the second statement.

Hence the agent finds the treasure before starting the execution of phase $\lceil \log_2 D \rceil +2$. By the first statement, the total  distance travelled by the agent during the first $\lceil \log_2 D \rceil +1$ phases is at most
$2^{(\lceil \log_2 D \rceil +2)+7}\leq 2^{10}*D$. Hence, the theorem holds.
\end{proof}

\section{Angles bounded by $\beta<2\pi$}

In this section we consider the case when all hints are angles upper-bounded by some constant $\beta<2\pi$, unknown to the agent. The main result of this section is Algorithm {\tt TreasureHunt2} whose cost is at most $O(D^{2-\epsilon})$, for some $\epsilon>0$. For a hint $(P_1,P_2)$ we denote by $\overline{(P_1,P_2)}$ the complement of $(P_1,P_2)$.

\subsection{High level idea}
\label{sub:hgi}
In Algorithm {\tt TreasureHunt2}, similarly as in the previous algorithm, 
  the agent acts in phases $j=1,2,3,\ldots$, where in each phase $j$ the agent ``supposes'' that the treasure is in the straight square centered at its initial position and of side length $2^j$. The intended goal is to search each supposed square at relatively low cost, and to ensure the discovery of the treasure by the time the agent finishes the first phase for which the initial supposed square contains the treasure. However, the similarity with the previous solution ends there: indeed, the hints that may now be less precise do not allow us to use the same strategy within a given phase. Hence we adopt a different approach that we outline below and that uses the following notion of tiling.
Given a square $S$ with side of length $x>0$, $Tiling(i)$ of $S$, for any non-negative integer $i$, is the partition of square $S$ into $4^i$ squares with side of length $\frac{x}{2^i}$. Each of these squares, called {\em tiles}, is closed, i.e., contains its border, and hence neighboring tiles overlap in the common border.

 Let us consider a simpler situation in which the angle of every hint $(P_1,P_2)$ is always equal to the bound $\beta$: the general case, when the angles may vary while being at most $\beta$, adds a level of technical complexity that is unnecessary to understand the intuition. In the considered situation, the angle of each excluded zone $\overline{(P_1,P_2)}$ is always the same as well. The following property holds in this case: there exists an integer $i_{\beta}$ such that for every square $S$ and every hint $(P_1,P_2)$ given at the center of $S$, at least one tile of $Tiling(i_{\beta})$ of $S$ belongs to the excluded zone $\overline{(P_1,P_2)}$. 


 In phase $j$, the agent performs $k$ steps: we will indicate later how the value of $k$ should be chosen.
 At the beginning of the phase, the entire square $S$ is white. 
 In the first step, the agent gets a hint $(P_1,P_2)$ at the center of $S$. By the above property, we know that $\overline{(P_1,P_2)}$ contains at least one tile of $Tiling(i_{\beta})$ of $S$, and we have the guarantee that such a tile cannot contain the treasure.  All points of all tiles included in $\overline{(P_1,P_2)}$ are painted black in the first step. This operation does not require any move, as painting is performed in the memory of the agent. As a result, at the end of the first step, each tile of $Tiling(i_{\beta})$ of $S$ is either black or white, in the following precise sense: a black tile is a tile all of whose points are black, and a white tile is a tile all of whose {\em interior} points are white.

In the second step, the agent repeats the painting procedure at a finer level. More precisely, the agent moves to the center of each white tile $t$ of $Tiling(i_{\beta})$ of $S$. When it gets a hint at the center of a white tile $t$, there is at least one tile of $Tiling(i_{\beta})$ of $t$ that can be excluded.  As in the first step,  all points of these excluded tiles are painted black. Note that a tile of $Tiling(i_{\beta})$ of $t$ is actually a tile of $Tiling(2i_{\beta})$ of $S$. Moreover, each tile of $Tiling(i_{\beta})$ of $S$ is made of exactly $4^{i_\beta}$ tiles of  $Tiling(2i_{\beta})$ of $S$. Hence, as depicted in Figure~\ref{fig:tiling}, the property we obtain at the end of the second step is as follows: each tile of  $Tiling(2i_{\beta})$ of $S$ is either black or white.

\begin{figure}[!htbp]
\begin{center}
  \begin{minipage}[t]{0.4\linewidth}
    \centering
	\includegraphics[width=0.6\textwidth]{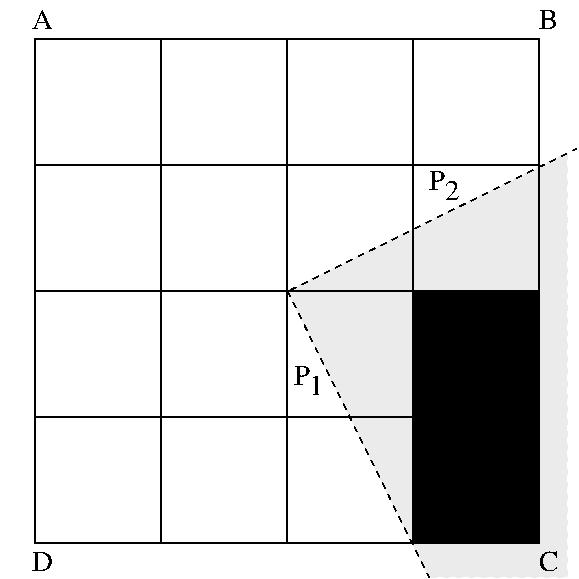}\\
    {\footnotesize ($a$) At the end of a first step\\for a hint $(P_1,P_2)$}
  \end{minipage}
  \begin{minipage}[t]{0.4\linewidth}
    \centering
	\includegraphics[width=0.6\textwidth]{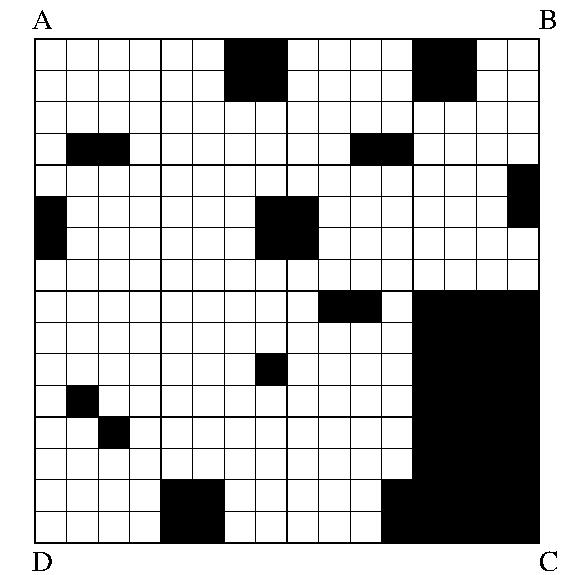}\\
    {\footnotesize ($b$) At the end of a second step}
  \end{minipage}
\end{center}
 \caption{White and black tiles at the end of the first and the second step of a phase, for square $S=ABCD$ and $i_\beta=2$.}
\label{fig:tiling}
\end{figure}


In the next steps, the agent applies a similar process at increasingly finer levels of tiling. More precisely, in step $2<s\leq k$, the agent moves to the center of each white tile of $Tiling((s-1)i_{\beta})$ of $S$ and gets a hint that allows it to paint black at least one tile of $Tiling(s\cdot i_{\beta})$ of $S$. At the end of step $s$, each tile of $Tiling(s\cdot i_{\beta})$ of $S$ is either black or white. We can show that at each step $s$ the agent paints black at least $\frac{1}{4^{i_{\beta}}}$th of the area of $S$ that is white at the beginning of step $s$.

After step $k$, each tile of $Tiling(k\cdot i_{\beta})$ of $S$ is either black or white.  These steps permit the agent to exclude some area without having to search it directly, while keeping some regularity of the shape of the black area. The agent paints black a smaller area than excluded by the hints but a more regular one. This regularity enables in turn the next process in the area remaining white. Indeed, the agent subsequently executes a brute-force searching that consists in moving to each white tile of $Tiling(k\cdot i_{\beta})$ of $S$ in order to scan it using the procedure {\tt RectangleScan}. If, after having scanned all the remaining white tiles, it has not found the treasure, the agent repaints white all the square $S$ and enters the next phase. Thus we have the guarantee that the agent finds the treasure by the end of phase $\lceil \log_2 D \rceil+1$, i.e., a phase in which the initial supposed square is large enough to contain the treasure. The question is: how much do we have to pay for all of this? In fact, the cost depends on the value that is assigned to $k$ in each phase $j$. The value of $k$ must be large enough so that the distance travelled by the agent during the brute-force searching is relatively small.  At the same time, this value must be small enough so that the the distance travelled during the $k$ steps is not too large. A good trade-off can be reached when $k=\lceil\log_{4^{i_\beta}}\sqrt{2^j}\rceil$. Indeed, as highlighted in the proof of correctness, it is due to this carefully chosen value of $k$ that we can beat the cost $\Theta(D^2)$ necessary without hints, and get a complexity of $\mathcal{O}(D^{2-\epsilon})$, where $\epsilon$ is a positive real depending on $i_\beta$, and hence depending on the angle $\beta$. 

\subsection{Algorithm and analysis}

In this subsection we describe our algorithm in detail, prove its correctness and analyze its complexity. We will use the notion of a slicing of a square.
Given a straight square $S$, the $Slicing(i)$ of $S$, for any integer $i\geq3$, is the partition of the square $S$ into $2^i$ triangles with a common vertex at the center $q$ of the square, resulting from partitioning the angle $2\pi$ into angles $\frac{2\pi}{2^i}$ using lines containing the point $q$, one of which is horizontal.

Consider any $Slicing(i)$ of a square $S$. Let $\Sigma$ be the set of all side lengths of triangles into which $Slicing(i)$ partitions $S$. We define $\rho_i$ to be the maximum of all integers $\lceil a/b \rceil$, where $a,b \in \Sigma$. Note that $\rho_i$ depends only on $i$ and not on the side length of $S$. Moreover, $\rho_{i+1}\geq\rho_i$.
%
%
For every integer $i\geq3$, we define $\phi(i)=i\rho_i$. 

In order to define some objects used by our algorithm, we need the following technical proposition.

\begin{proposition}
\label{pro:index}
The following properties hold.
\begin{enumerate}
\item For every angle $0<\alpha<2\pi$ with vertex at the center of a square $S$, the angle $\alpha$ contains some triangle of $Slicing(\max(3,\lceil\log_2(\frac{2\pi}{\alpha})\rceil+1))$ of $S$.
\item For every integer $i\geq3$ and for every triangle $T$ of $Slicing(i)$ of a square $S$, at least one tile of $Tiling(4\phi(i))$ of $S$ is included in the interior of $T$. \end{enumerate}
\end{proposition}

\begin{proof}
We start by proving the first property. Let $S$ be a square and let $0<\alpha<2\pi$ be an angle with the vertex in the center of $S$.
Let $i=max(3; \lceil\log_2(\frac{2\pi}{\alpha})\rceil+1)\geq3$.
The angle at the center of square $S$ in each of the triangles of $Slicing(i)$ of $S$ is at most $\frac{\alpha}{2}$. 
Hence one of the triangles formed by $Slicing(i)$ is included in the angle $\alpha$. This proves the first property.


In the proof of the second property, all tilings and slicings are for square $S$: for ease of reading we omit mentioning it.
In order to prove the second property, we first prove by induction on $i$ the following statement denoted by $\mathcal{H}_i$:\\ 
For every integer $i\geq3$ and for every triangle $T$ of $Slicing(i)$, there is at least one tile $t$ of $Tiling(4\phi(i)-2)$, such that $t\subset T$ and one side of $t$ is included in a side of $S$.


For the base case $i=3$, note that each triangle of $Slicing(3)$ contains at least one tile of $Tiling(2)$ with one side  included in a side of $S$. Since $\phi(3)=3\rho_3$ and $\rho_3\geq1$, we know that $4\phi(3)-2\geq10$. Moreover,  each side of every tile $t'$ of $Tiling(r)$ contains at least one side of a tile of $Tiling(r')$, included in $t'$, for all integers $r<r'$. Hence $\mathcal{H}_3$ is true. 

Assume that $\mathcal{H}_j$ is true for some integer $j\geq3$ and let us prove that  $\mathcal{H}_{j+1}$ is also true. Suppose by contradiction that $\mathcal{H}_{j+1}$ is false. This means that there exists a triangle $T_1$ of $Slicing(j+1)$ that contains no tile of $Tiling(4\phi(j+1)-2)$ with one side included in a side of $S$. 
Denote by $L$ the side of $S$ that contains a side of $T_1$. There exists a triangle $T$ of $Slicing(j)$ and a triangle $T_2$ of $Slicing(j+1)$ such that $T_1\cup T_2=T$ and $T_1\cap T_2=l$, where $l$ is the common segment of boundaries of $T_1$ and $T_2$. Note that triangle $T_2$ also has a side included in $L$.

By the inductive hypothesis, there exists a tile $t'$ of $Tiling(4\phi(j)-2)$ such that $t'\subset T$ and one side of $t'$ is included in $L$. For any integers $r<r'$, every tile of $Tiling(r)$ contains exactly $4^{r'-r}$ tiles of $Tiling(r')$ that are organized in $2^{r'-r}$ rows of $2^{r'-r}$ squares. So, tile $t'$ contains exactly $4^{2\phi(j+1)-2\phi(j)}$ rows that are parallel to $L$ and such that each of them is made of $4^{2\phi(j+1)-2\phi(j)}$ tiles of $Tiling(4\phi(j+1)-2)$. Among these rows consider the one that has a common boundary with $L$ and denote it by $R$. Note that $R$ contains at least $4^{2\rho_{j+1}}$ tiles of $Tiling(4\phi(j+1)-2)$ because $2\phi(j+1)-2\phi(j)=2(j+1)\rho_{j+1}-2j\rho_{j}$ and $\rho_{j+1}\geq\rho_{j}$. Denote by $R'$ the row of $Tiling(4\phi(j+1)-2)$ that contains $R$ and by $R''$ the part of $R'$ made of tiles $t''$ of $Tiling(4\phi(j+1)-2)$, such that $t''\subset T$. Note that $R\subseteq R''$ and thus $R''$ contains at least $4^{2\rho_{j+1}}$ tiles of $Tiling(4\phi(j+1)-2)$. Moreover, note that the smaller of the two angles formed by $l$ and $L$ cannot be smaller than $\frac{\pi}{4}$ or larger than $\frac{\pi}{2}$. As a result, $l$ can intersect at most $2$ adjacent tiles $s_1,s_2$ of $R''$.  We will show that $l$ cannot intersect a tile that is at an end of row $R''$. Let $x$ be the side length of a tile of $Tiling(4\phi(j+1)-2)$.
Suppose that $l$ intersects a tile that is at an end of row $R''$.
In view of the fact that $R''$ contains all tiles of $R$ that are included in $T$, a side of $T_1$ or of $T_2$ included in $L$ (say the side of $T_1$ without loss of generality), has length at most $3x$, while the side of $T_2$ included in  $L$ has  length at least $(4^{2\rho_{j+1}}-2)x\geq 14\rho_{j+1}x$.  However, $\frac{14\rho_{j+1}x}{3x}>\rho_{j+1}$, which contradicts the definition of $\rho_{j+1}$. Hence $l$ cannot intersect a tile that is at an end of row $R''$.
This implies that one of the two tiles at the ends of $R''$ belongs to $T_1$: by construction this tile belongs to $Tiling(4\phi(j+1)-2)$ with one side belonging to $L$. Hence we get a contradiction. As a result, $\mathcal{H}_{j+1}$ is true, which ends the proof by induction of $\mathcal{H}_{i}$.

It remains to conclude the proof of the second property of our proposition. In view of property $\mathcal{H}_i$, we know that for every integer $i\geq3$ and for every triangle $T$ of $Slicing(i)$, at least one tile of $Tiling(4\phi(i)-2)$ is included in $T$. Moreover, each tile of $Tiling(4\phi(i)-2)$ contains $4$ rows, each made of $4$ tiles belonging to $Tiling(4\phi(i))$. Hence, the interior of each tile of $Tiling(4\phi(i)-2)$ contains a tile of $Tiling(4\phi(i))$. This proves the second property and concludes the proof of the proposition.
\end{proof}

For any angle $0<\alpha<2\pi$, the index of $\alpha$, denoted $index(\alpha)$, is the integer $4\phi(max(3,\lceil\log_2(\frac{2\pi}{\alpha})\rceil+1))$.
Proposition \ref{pro:index} implies


\begin{proposition}
\label{pro:index2}
For every angle $0<\alpha<2\pi$, the following properties hold.

\begin{enumerate}
\item
For every square $S$ and for every hint $(P_1,P_2)$ of size $2\pi-\alpha$ obtained at the center of $S$, there exists a tile of $Tiling(index(\alpha))$ of $S$ included in $\overline{(P_1,P_2)}$.
\item
 For every angle $\alpha'<\alpha$, we have  $index({\alpha}) \leq index({\alpha'}) $.
\end{enumerate}
\end{proposition}


Algorithm \ref{alg:obtus} gives a pseudo-code of the main algorithm of this section. It uses the function {\tt Mosaic} described in Algorithm \ref{alg:mosaic}  that is the key technical tool permitting the agent to reduce its search area. The agent interrupts the execution of Algorithm  \ref{alg:obtus} as soon as it gets at distance 1 from the treasure, at which point it can ``see'' it and thus treasure hunt stops.

\begin{algorithm}
				\caption{{\tt TreasureHunt2}}
				\label{alg:obtus}
				\begin{footnotesize}
				\begin{algorithmic}[1]
					\State $IndexNew:=1$
					\State $i:=1$
					\Loop 
    					\Repeat
						\State $IndexOld:= IndexNew$\label{obt:1}
						\State $IndexNew:=$ {\tt Mosaic}$(i,IndexOld)$\label{obt:2}
					\Until{$IndexNew=IndexOld$}
					\State $ i:=i+1$
					\EndLoop				
				\end{algorithmic}
				\end{footnotesize}
			\end{algorithm}
			
			In the following, a square is called black if all its points are black. A square is called white if all points of its interior are white. (In a white square, some points
of its border may be black).

\begin{algorithm}
				\caption{Function {\tt Mosaic}($i$,$k$)}
				\label{alg:mosaic}
				\begin{footnotesize}
				\begin{algorithmic}[1]
					\State $O$:= the initial position of the agent \label{l1:debut}
					\State $S$:= the straight square centered at $O$ with sides of length $2^i$ 
					\State Paint white all points of $S$
					\State $IndexMax$:=$k$ \label{l:assign}
					\For{$j=1$ {\bf to $\lceil\log_{4^k}\sqrt{2^i}\rceil$}}\label{loop:debut}
						\ForAll{tiles $t$ of $Tiling((j-1)k)$ of $S$}
							\If{$t$ is white}
							\State Go to the center of $t$
							\State Let $(P_1,P_2)$ be the obtained hint
							\State $k'$:= index of $\overline{(P_1,P_2)}$
							\If{$k'>IndexMax$} \label{l:proof}
								\State $IndexMax$:=$k'$\label{l:proof2}
							\EndIf
							\If{$IndexMax = k$}
								\ForAll{tiles $t'$ of $Tiling(k)$ of $t$ such that $t'\subset\overline{(P_1,P_2)}$}\label{line:pb1}
								\State Paint black all points of $t'$\label{line:pb}
								\EndFor\label{line:pb2}
							\EndIf
							\EndIf
						\EndFor	\label{loop:fin}
					\EndFor \label{l1:fin}
					\If{$IndexMax=k$} \label{l2:debut}
						\ForAll{tiles $t$ of $Tiling(k(\lceil\log_{4^k}\sqrt{2^i}\rceil))$ of $S$}\label{b2:debut}
							\If{$t$ is white}
								\State Go to the center of $t$
								\State Execute {\tt RectangleScan}($t$)
							\EndIf
						\EndFor\label{b2:fin}
					\EndIf
					\State Go to $O$ \label{l2:avfin}
					\State \Return $IndexMax$	\label{l2:fin}	
				\end{algorithmic}
				\end{footnotesize}
			\end{algorithm}




\begin{lemma}
\label{lem:ob1}
For any positive integers $i$ and $k$,
consider an agent executing function {\tt Mosaic}($i$,$k$) from its initial position $O$. Let $S$ be the straight square centered at $O$ with side of length $2^i$. For every positive integer $j\leq\lceil\log_{4^k}\sqrt{2^i}\rceil$, at the end of the $j$-th execution of the first loop (lines~\ref{loop:debut} to~\ref{loop:fin}) in {\tt Mosaic}($i$,$k$), each tile of $Tiling(jk)$ of $S$ is either black or white. 
\end{lemma}

\begin{proof}
Assume by contradiction that there exists a positive integer $j\leq\lceil\log_{4^k}\sqrt{2^i}\rceil$ such that at the end of the $j$-th execution of the first loop, there exists at least one tile $\sigma$ of $Tiling(jk)$ of $S$ that is neither black nor white. Without loss of generality, we assume that $j$ is the first integer for which this occurs.

In view of the minimality of $j$, we know that just before starting the $j$-th execution of  the first loop, each tile of $Tiling((j-1)k)$ is either black or white. Moreover, for every positive integers $z'\leq z$, every couple of points that belong to the same tile of $Tiling(z)$ of $S$, also belong to the same tile of  the coarser tiling $Tiling(z')$ of $S$. Hence, just before starting the $j$-th execution of the first loop, each tile of $Tiling(jk)$ is either black or white. 

During the execution of the first loop, the points that become black remain always black thereafter. Since there exists a tile $\sigma$ of $Tiling(jk)$ of $S$ that becomes neither black nor white during the $j$-th execution of the first loop, at some point during this execution, the agent does not paint black all points of $\sigma$ when executing line~\ref{line:pb} of Algorithm~\ref{alg:mosaic}. However, each time the agent executes line~\ref{line:pb} of Algorithm~\ref{alg:mosaic} within the $j$-th execution of the first loop, when a point of a tile $t'$ of $Tiling(k)$ of any tile of $Tiling((j-1)k)$ of $S$ is painted black, then all points inside and on the boundary of tile $t'$ are painted black. By definition, $t'$ is a tile of $Tiling(jk)$ of $S$. Hence, at the end of the $j$-th execution of the first loop, each tile of $Tiling(jk)$ of $S$ is either black or white. Hence, we get a contradiction with the existence of $\sigma$ which proves the lemma.
\end{proof}

\begin{lemma}
\label{lem:ob2}
For every positive integers $i$ and $k$, a call to function {\tt Mosaic}($i$,$k$) has cost at most $2^{i\frac{3+\log_{4^k}(4^k-1)}{2}+2k+8}$.
\end{lemma}

\begin{proof}
The walk made by the agent executing function {\tt Mosaic}($i$,$k$) can be divided into two parts: the first part $\mathcal{P}_1$ is the walk made by executing lines~\ref{l1:debut} to~\ref{l1:fin} of Algorithm~\ref{alg:mosaic}, while the second part $\mathcal{P}_2$ is the walk made by executing lines~\ref{l2:debut} to~\ref{l2:fin} of Algorithm~\ref{alg:mosaic}. The distance traveled in $\mathcal{P}_1$ (resp. $\mathcal{P}_2$) will be denoted by $|\mathcal{P}_1|$ (resp. $|\mathcal{P}_2|$). We first focus on the distance traveled in part $\mathcal{P}_1$, in which the walk made by the agent is as follows: for each $j\in\{1,\ldots \lceil\log_{4^k}\sqrt{2^i}\rceil\}$, starting from the center of $S$, the agent moves to the center of every white tile of $Tiling((j-1)k)$ of $S$. By Algorithm~\ref{alg:mosaic}, the side length of $S$ is $2^i$, and thus the distance between any two points of $S$ is upper bounded by $2^{i+1}$. Moreover, if for every non-negative integer $s$, we denote by $\mathcal{Q}_s$ the number of tiles in $Tiling(s)$ of $S$, then we have

\begin{align}
|\mathcal{P}_1| &\leq 2^{i+1}\sum_{j=1}^{\lceil\log_{4^k}\sqrt{2^i}\rceil}\mathcal{Q}_{(j-1)k}\label{form:1}
\end{align}

In view of the definition of a tiling,  for all $j\in\{1,\ldots \lceil\log_{4^k}\sqrt{2^i}\rceil\}$ we have
\begin{align}
\label{form:2}
\mathcal{Q}_{(j-1)k} &= \frac{\mathcal{Q}_{(\lceil\log_{4^k}\sqrt{2^i}\rceil-1)k}}{4^{(\lceil\log_{4^k}\sqrt{2^i}\rceil-1)k-(j-1)k}}
\end{align}
\begin{align}
\label{form:3}
~&= \frac{\mathcal{Q}_{(\lceil\log_{4^k}\sqrt{2^i}\rceil-1)k}}{4^{(\lceil\log_{4^k}\sqrt{2^i}\rceil-j)k}}
\end{align}

Hence, in view of (\ref{form:1}) and (\ref{form:3}), we have 

\begin{align}
\label{form:4}
|\mathcal{P}_1| &\leq 2^{i+1}\sum_{j=1}^{\lceil\log_{4^k}\sqrt{2^i}\rceil}\frac{\mathcal{Q}_{(\lceil\log_{4^k}\sqrt{2^i}\rceil-1)k}}{4^{(\lceil\log_{4^k}\sqrt{2^i}\rceil-j)k}}\\
~&\leq 2^{i+2}\mathcal{Q}_{(\lceil\log_{4^k}\sqrt{2^i}\rceil-1)k}\label{form:5}
\end{align}

In view of the definition of a tiling, we have 

\begin{align}
\label{form:6}
\mathcal{Q}_{(\lceil\log_{4^k}\sqrt{2^i}\rceil-1)k} &= 4^{(\lceil\log_{4^k}\sqrt{2^i}\rceil-1)k}\\
~&= (4^k)^{\lceil\log_{4^k}\sqrt{2^i}\rceil-1}\\
~&\leq \sqrt{2^i}\label{form:9}
\end{align}

Hence from (\ref{form:5}) and (\ref{form:9}), we obtain

\begin{align}
\label{form:10}
|\mathcal{P}_1| &\leq 2^{\frac{3i}{2}+2}
\end{align}

We now consider the distance traveled in part $\mathcal{P}_2$. Here, there are two cases: either $IndexMax\ne k$ when the agent starts executing line~\ref{l2:debut} of Algorithm~\ref{alg:mosaic}, or $IndexMax=k$ when the agent starts executing line~\ref{l2:debut} of Algorithm~\ref{alg:mosaic}. In the first case, $\mathcal{P}_2$ corresponds only to the move made when executing line~\ref{l2:avfin} of Algorithm~\ref{alg:mosaic}. However, during the entire execution of Algorithm~\ref{alg:mosaic}, the agent never leaves the straight square $S$, centered at $O$, whose sides have length $2^i$. Hence in the first case, $|\mathcal{P}_2|\leq2^{i+1}$.

The second case is trickier to analyze. Indeed, we have to take into account the distance traveled when executing line~\ref{l2:avfin} of Algorithm~\ref{alg:mosaic} (that is upper bounded by $2^{i+1}$ in this case as well) but also the distance traveled when executing lines~\ref{b2:debut} to~\ref{b2:fin}: note that since those lines are executed, we necessarily have the following claim in the second case.

\begin{claim}
\label{cl:1}
Once variable IndexMax is assigned the value $k$ (cf. line~\ref{l:assign} of Algorithm~\ref{alg:mosaic}), variable IndexMax does not change anymore thereafter.
\end{claim}
 
The above claim is used in the proof of the following one that is crucial to determine the traveled distance $|\mathcal{P}_2|$. As for Claim~\ref{cl:1}, Claim~\ref{cl:2} holds in the second case that we currently analyze.

\begin{claim}
\label{cl:2}
At the end of part $\mathcal{P}_1$, the area of the white surface is at most $2^{i\frac{3+\log_{4^k}(4^k-1)}{2}}$. 
\end{claim}

\begin{proofclaim}
To prove the claim, we first show by induction on $j$ the following property $\mathcal{K}_j$:\\ 
For every integer $j\in\{1,\ldots,\lceil\log_{4^k}\sqrt{2^i}\rceil\}$, at the end of the $j$-th execution of the first loop of Algorithm~\ref{alg:mosaic} the area of the part of the square $S$ that is still white is at most $(\frac{4^k-1}{4^k})^j2^{2i}$.

During the first execution of the first loop of Algorithm~\ref{alg:mosaic}, the agent is located at the center of $S$. By Claim~\ref{cl:1}, the agent executes line~\ref{line:pb} during this first execution, and by Proposition~\ref{pro:index2}, there is at least one tile $t'$ of $Tiling(k)$ of $S$ such that all points of $t'$ are black. Since there are $4^k$ tiles in $Tiling(k)$ of $S$, it follows that property $\mathcal{K}_j$ is true for $j=1$. Now suppose that property $\mathcal{K}_s$  holds for a positive integer $s$. We  show that $\mathcal{K}_{s+1}$ is also true. It is enough to show that at the end of the $(s+1)$-th execution of the first loop of Algorithm~\ref{alg:mosaic} the part of the square $S$ that is still white has area at most $(\frac{4^k-1}{4^k})^{s+1}2^{2i}$. In view of Claim~\ref{cl:1} and Algorithm~\ref{alg:mosaic}, during this $(s+1)$-th execution the agent goes to the center of every white tile of $Tiling(sk)$ of $S$ from which it executes line~\ref{line:pb} of Algorithm~\ref{alg:mosaic}. Moreover, by Claim~\ref{cl:1}, we know that the value of variable $k'$ is never larger than $k$. Hence, by Proposition~\ref{pro:index2}, it follows that the agent paints black at least $(\frac{1}{4^k})$-th of each white tile of $Tiling(sk)$ of $S$ during this $(s+1)$-th execution. However, at the beginning of the $(s+1)$-th execution of the first loop, we know from the inductive hypothesis and from Lemma~\ref{lem:ob1}, that the sum of the areas of the white tiles of $Tiling(sk)$ is at most $(\frac{4^k-1}{4^k})^s2^{2i}$. Moreover, by painting black at least $(\frac{1}{4^k})$-th of each white tile of $Tiling(sk)$ of $S$, the agent paints black at least $(\frac{1}{4^k})$-th of the remaining surface that is white at the beginning of the $(s+1)$-th execution of the first loop. This implies $\mathcal{K}_{s+1}$, which concludes the proof by induction of  $\mathcal{K}_{j}$.

From property $\mathcal{K}_j$ with $j\in\{1,\ldots,\lceil\log_{4^k}\sqrt{2^i}\rceil\}$, we know that at the end of part $\mathcal{P}_1$, the area of the white surface is at most 

\begin{align}
\label{form:11}
2^{2i}(\frac{4^k-1}{4^k})^{\lceil\log_{4^k}\sqrt{2^i}\rceil}\leq 2^{2i}(\frac{4^k-1}{4^k})^{\log_{4^k}\sqrt{2^i}}
\end{align}

However, we have 

\begin{align}
\label{form:12}
(\frac{4^k-1}{4^k})^{\log_{\frac{4^k-1}{4^k}}\sqrt{2^i}}= \sqrt{2^i}
\end{align}

which implies 

\begin{align}
\label{form:13}
(\frac{4^k-1}{4^k})^{\log_{4^k}\sqrt{2^i}}= 2^{\frac{i}{2}\log_{4^k}(\frac{4^k-1}{4^k})}.
\end{align}

It follows from (\ref{form:11}) and (\ref{form:13}) that the area of the white surface at the end of part $\mathcal{P}_1$ is at most 

\begin{align}
\label{form:14}
2^{2i+\frac{i}{2}\log_{4^k}(\frac{4^k-1}{4^k})}=2^{i\frac{3+\log_{4^k}(4^k-1)}{2}},
\end{align}

which concludes the proof of the claim.
\end{proofclaim}

Now, we are ready to compute $|\mathcal{P}_2|$ in the case where the condition $IndexMax=k$ holds when the agent executes line~\ref{l2:debut} of Algorithm~\ref{alg:mosaic}. The value of $|\mathcal{P}_2|$ is the sum of the distance traveled when executing line~\ref{l2:avfin} (upper bounded by $2^{i+1}$) and of the distance traveled when executing lines~\ref{b2:debut} to~\ref{b2:fin}. When executing the latter block of lines, for each white tile $t$ of $Tiling(k(\lceil\log_{4^k}\sqrt{2^i}\rceil))$ of $S$, the agent performs successively the two following actions:
\begin{enumerate}
\item The agent moves to the center of $t$, at a cost of at most $2^{i+1}$. 
\item Once the center of $t$ is reached, the agent executes procedure {\tt RectangleScan}($t$), at a cost of at most $5l \cdot max(2,l)$ (cf. Proposition~\ref{prelim}) with $l$ equal to the side length of tile $t$.
\end{enumerate}
Hence, if we denote by $w$ the number of white tiles in $Tiling(k(\lceil\log_{4^k}\sqrt{2^i}\rceil))$ of $S$, we have 

\begin{align}
\label{form:15}
|\mathcal{P}_2|&\leq 2^{i+1}+w(2^{i+1}+5l\cdot max(2,l))\\
~&\leq 2^{i+1}(w+1)+8w\cdot l \cdot max(2,l)\\
~&\leq 2^{i+1}(w+1)+8w \cdot l^2+{32}w \label{form:17}
\end{align}

By the definition of tiling we have $w\leq 4^{k(\lceil\log_{4^k}\sqrt{2^i}\rceil)}\leq 2^{{2}k+\frac{i}{2}}$. Moreover, in view of Claim~\ref{cl:2} and Lemma~\ref{lem:ob1}, we know that $w\cdot l^2\leq 2^{i\frac{3+\log_{4^k}(4^k-1)}{2}}$. Thus, from (\ref{form:17}) we have the following:

\begin{align}
\label{form:18}
|\mathcal{P}_2|&\leq 2^{{2}k+\frac{3i}{2}+1}+2^{i+1}+2^{i\frac{3+\log_{4^k}(4^k-1)}{2}+3}+2^{2k+\frac{i}{2}+{5}}\\
~&\leq 2^{i\frac{3+\log_{4^k}(4^k-1)}{2}+2k+7} \mbox{~~~~~~}(because\mbox{~}k\geq1).
\end{align}
So, whether $IndexMax= k$ or not when the agent starts executing line~\ref{l2:debut} of Algorithm~\ref{alg:mosaic}, we have $|\mathcal{P}_2|\leq 2^{i\frac{3+\log_{4^k}(4^k-1)}{2}+2k+{7}}$. Hence,
$|\mathcal{P}_1|+|\mathcal{P}_2|\leq 2^{\frac{3i}{2}+2}+ 2^{i\frac{3+\log_{4^k}(4^k-1)}{2}+2k+7}\leq2^{i\frac{3+\log_{4^k}(4^k-1)}{2}+2k+8}$, which concludes the proof of the lemma.
\end{proof}

Let $\psi$ be the index of $2 \pi - \beta$. The next proposition follows from Proposition~\ref{pro:index2}.

\begin{proposition}
\label{pro:maxin}
Let $(P_1,P_2)$ be any hint. The index of $\overline{(P_1,P_2)}$ is at most $\psi$. 
\end{proposition}

We are now ready to prove the final result of this section.

\begin{theorem}
Consider an agent $A$ and a treasure located at distance at most $D$ from the initial position of $A$. By executing Algorithm {\tt TreasureHunt2}, agent $A$ finds the treasure after having traveled a distance in $\mathcal{O}(D^{2-\epsilon})$, for some $\epsilon>0$.
\end{theorem}

\begin{proof}
We will use the following two claims.

\begin{claim}
\label{cl:3}

Let $i\geq 1$ be an integer. The number of executions of the repeat loop in the $i$-th execution of the external loop in Algorithm ~\ref{alg:obtus} is bounded by $\psi$.
\end{claim}

%


\begin{proofclaim}
Suppose by contradiction that the claim does not hold for some $i\geq 1$. So, the number of executions of the repeat loop in the $i$-th execution of the external loop in Algorithm~\ref{alg:obtus} is at least $\psi+1$. In each of these executions of the repeat loop, the agent calls function {\tt Mosaic}$(i,*)$ exactly once. For all $1\leq j\leq \psi+1$ ($\psi\geq1$, by definition of an index), denote by $v_j$ the returned value of function  {\tt Mosaic}$(i,*)$ in the $j$-th execution of the repeat loop in the $i$-th execution of the external loop. Note that  $v_1\ne1$: indeed, if $v_1=1$ the repeat loop would be executed exactly once, which would be a contradiction because it is executed at least $\psi+1\geq2$ times.

In view of Algorithm~\ref{alg:obtus} and Proposition~\ref{pro:maxin}, the returned value of {\tt Mosaic}$(i,*)$ is a positive integer that is at most $\psi$. Since $v_1\ne1$, this implies that $\psi\geq2$. Moreover, for all $2\leq j\leq \psi$, we have $v_{j}\geq v_{j-1}$ (cf. lines~\ref{obt:1}-\ref{obt:2} of Algorithm~\ref{alg:obtus} and lines~\ref{l:assign},~\ref{l:proof}-\ref{l:proof2} of Algorithm~\ref{alg:mosaic}). Hence, there exists an integer $k\leq\psi$ such that  $v_{k}=v_{k-1}$. However, according to Algorithm~\ref{alg:obtus}, this implies that the number of executions of the repeat loop in the $i$-th execution of the external loop is at most $k\leq\psi$.  This is a contradiction which concludes the proof of the claim.
\end{proofclaim}

\begin{claim}
\label{cl:4}
The distance traveled by the agent before variable $i$ becomes equal to $\lceil\log_2D\rceil+2$ in the execution of  Algorithm~\ref{alg:obtus} is 
 $\mathcal{O}(D^{2-\epsilon})$, where $\epsilon=\frac{1}{2}(1-\log_{4^\psi}(4^\psi-1))>0$.
\end{claim}
\begin{proofclaim}
In view of the fact that the returned value of every call to function {\tt Mosaic} in the execution of Algorithm~\ref{alg:obtus} is at most $\psi$, it follows that in each call to function {\tt Mosaic}$(*,k)$ the parameter $k$ is always at most $\psi$. Hence, in view of Claim \ref{cl:3} and Lemma~\ref{lem:ob2}, as long as variable $i$ does not reach the value $\lceil\log_2D\rceil+2$, the agent traveled a distance at most 


\begin{align}
\label{form:19}
~&\psi \cdot \sum_{i=1}^{\lceil\log_2D\rceil+1} 2^{i\frac{3+\log_{4^\psi}(4^\psi-1)}{2}+2\psi+8}\\
\leq &\psi 2^{{(\lceil\log_2D\rceil+1)}\frac{3+\log_{4^\psi}(4^\psi-1)}{2}+2\psi+{9}}\\
\leq &\psi 2^{2\psi+{12}+\log_{4^\psi}(4^\psi-1)}2^{(\log_2D)\frac{3+\log_{4^\psi}(4^\psi-1)}{2}}\\
= &\psi 2^{2\psi+{12}+\log_{4^\psi}(4^\psi-1)}D^{2-\frac{1}{2}(1-\log_{4^\psi}(4^\psi-1))}\label{form:22}
\end{align}

By (\ref{form:22}), the total distance traveled by the agent executing Algorithm~\ref{alg:obtus} is $\mathcal{O}(D^{2-\epsilon})$ where $\epsilon=\frac{1}{2}(1-\log_{4^\psi}(4^\psi-1))$. Since $\psi$ is a positive integer, we have $0<\log_{4^\psi}(4^\psi-1)<1$ and hence $\epsilon>0$. This ends the proof of the claim.
\end{proofclaim}

Assume that the theorem is false. As long as variable $i$ does not reach $\lceil\log_2D\rceil+2$, the agent cannot find the treasure, as this would contradict Claim~\ref{cl:4}. Thus, in view of Claim~\ref{cl:3}, before the time $\tau$ when variable $i$ reaches $\lceil\log_2D\rceil+2$ the treasure is not found. By Algorithm~\ref{alg:obtus}, this implies that during the last call to function {\tt Mosaic} before time $\tau$, the function returns a value that is equal to its second input parameter. This implies that during this call, the agent has executed lines~\ref{b2:debut} to~\ref{b2:fin} of Algorithm~\ref{alg:mosaic}: more precisely, there is some integer $x$ such that from each white tile $t$ of $Tiling(x)$ of the straight square $S$ that is centered at the initial position of the agent and that has sides of length $2^{\lceil\log_2D\rceil+1}$, the agent has executed function {\tt RectangleScan}$(t)$. Hence, at the end of the execution of lines~\ref{b2:debut} to~\ref{b2:fin}, the agent has seen all points of each white tile of $Tiling(x)$ of $S$. Moreover, in view of Lemma~\ref{lem:ob1}, we know that the tiles that are not white, in $Tiling(x)$ of $S$, are necessarily black. Given a black tile $\sigma$ of $Tiling(x)$, each point of $\sigma$ is black, which, in view of lines~\ref{line:pb1} to~\ref{line:pb2} of Algorithm~\ref{alg:mosaic}, implies that $\sigma$ cannot contain the treasure. Since square $S$ necessarily contains the treasure, it follows that the agent must find the treasure by the end of the last execution of function {\tt Mosaic} before time $\tau$. As a consequence, the agent stops the execution of Algorithm~\ref{alg:obtus} before assigning $\lceil\log_2D\rceil+2$ to variable $i$ and thus, we get a contradiction with the definition of time $\tau$, which proves the theorem.
\end{proof}

\section{Arbitrary angles}

In this section we observe that if hints can be arbitrary angles smaller than $2\pi$ then the treasure hunt cost $\Theta(D^2)$ cannot be improved in the worst case.
We prove the following proposition.

\begin{proposition}\label{lower}
If hints can be arbitrary angles smaller than $2\pi$ then the optimal cost of treasure hunt for a treasure at distance at most $D$ from the starting point of the agent is $\Omega(D^2)$.
\end{proposition}

\begin{proof}
Consider the disc $\cD$ of radius $D$ centered at the initial point of the agent. Consider any position of the agent and suppose that the angle given as hint has size $\gamma >\pi$. Call the complement of the hint the {\em forbidden angle}. It has size $\alpha=2\pi-\gamma<\pi$. The forbidden angle has the property that the treasure must be outside of it. If the current position of the agent is outside of disc $\cD$ then the forbidden angle of size $\alpha$ can be chosen in such a way that it is disjoint from $\cD$, i.e., that it does not exclude any point of the disc $\cD$ as a possible location of the treasure. If the current position of the agent is in $\cD$, then the intersection of the forbidden angle with $\cD$ has area at most $\frac{\alpha}{2\pi} \pi (2D)^2=\alpha \cdot 2D^2$, as it is at most the area of a sector of the disc with radius $2D$ and angle $\alpha$.

Suppose that there exists a treasure hunt algorithm at cost at most $D^2/2$. Let the sizes of forbidden angles corresponding to consecutive hints be $\frac{1}{2}$, $\frac{1}{4}$, $\frac{1}{8},...$ etc., each of size half of the preceding, and such that the forbidden angle is disjoint from the disc $\cD$, whenever the current position of the agent is outside of $\cD$. When the position is in $\cD$, the angle of the respective size can be chosen arbitrarily. The total area of the intersection of $\cD$ with the forbidden angles is at most  $(\sum_{i=1}^{\infty} \frac{1}{2^i})\cdot 2D^2= 2D^2$. This leaves out a part of the disc $\cD$ whose area is at least $(\pi-2)D^2$. During the walk of length at most $D^2/2$ of the agent,  the set of points of $\cD$ from which the agent is at distance at most 1 at some point of the walk has area at most $\frac{D^2}{2}\cdot 2+\pi=D^2+\pi$. For $D>5$ we have $(\pi-2)D^2>D^2+\pi$.   
Hence there exists a point of $\cD$ not included in any of the forbidden angles, from which the agent has never been at distance at most 1. 
Placing the treasure in this point refutes the correctness of the treasure hunt algorithm. This implies that the trajectory of the agent must have length larger than $D^2/2$,
for $D>5$,
hence the optimal cost of treasure hunt is $\Omega(D^2)$.
\end{proof}

\section{Conclusion}

For hints that are angles at most $\pi$ we gave a treasure hunt algorithm with optimal cost linear in $D$. For larger angles we showed a separation between the case where angles are bounded away from $2\pi$, when we designed an algorithm with cost strictly subquadratic in $D$, and the case where angles have arbitrary values smaller than $2\pi$, when
we showed a quadratic lower bound on the cost. The optimal cost of treasure hunt with large angles bounded away from $2\pi$ remains open. In particular, the following questions seem intriguing. Is the optimal cost linear in $D$ in this case, or is it possible to prove a super-linear lower bound on it? Does the order of magnitude of this optimal cost depend on the bound $\pi<\beta<2\pi$ on the angles given as hints?

\bibliographystyle{plain}


\end{document}